\documentclass{article}
\usepackage{ijcai25}

\title{On Minimal Achievable Quotas in Multiwinner Voting}

\usepackage{tabularx}
\usepackage{balance}
\usepackage{todonotes}
\usepackage{tikz}
\usepackage{amsmath, amsthm, amssymb}
\usepackage[noend]{algpseudocode}
\usepackage{algorithm}
\usepackage{placeins}
\usetikzlibrary{calc,math}
\usepackage{thmtools}
\usepackage{thm-restate}
\usepackage[noend]{algpseudocode}
\usepackage{algorithm}
\usepackage{graphicx}
\usepackage{csquotes}
\usepackage{natbib}
\usepackage{cleveref}
\usepackage{url}
\usepackage{booktabs}
\usepackage[T1]{fontenc}   
\usepackage[utf8]{inputenc}  
\usepackage{textcomp}   
\usepackage{float}
\usepackage{caption}
\DeclareMathOperator{\Q}{\mathbb{Q}}
\DeclareMathOperator{\N}{\mathbb{N}}
\DeclareMathOperator{\R}{\mathbb{R}}

\newtheorem{theorem}{Theorem}[section]
\newtheorem{lemma}[theorem]{Lemma}
\newtheorem{proposition}[theorem]{Proposition}

\theoremstyle{definition}
\newtheorem{definition}[theorem]{Definition}
\newtheorem{example}[theorem]{Example}

\author{Patrick Becker \and Fabian Frank
\affiliations{
Technical University of Munich
}}

\begin{document}
\maketitle 

\begin{abstract}
    Justified representation (JR) and extended justified representation (EJR) are well-established proportionality axioms in approval-based multiwinner voting. 
    Both axioms are always satisfiable, but they rely on a fixed quota (typically Hare or Droop), with the Droop quota being the smallest one that guarantees existence across all instances. 
    With this in mind, we take a step beyond the fixed-quota paradigm by studying instance-dependent proportionality notions. 
    More specifically, we minimize the quota requirements for JR and EJR using the parameter $\alpha$.
    We demonstrate that all commonly studied voting rules can have an additive gap to the optimum of $\frac{k^2}{(k+1)^2}$. 
    Moreover, we examine the computational aspects of our instance-dependent quota and prove that determining the optimal value of $\alpha$ for a given approval profile that allows some committee to satisfy $\alpha$-JR is NP-complete. 
    To address this, we introduce an integer linear programming (ILP) formulation for computing committees that satisfy $\alpha$-JR, and we provide positive computational results in the voter interval (VI) and candidate interval (CI) domains.
\end{abstract}

\section{Introduction}
In the context of multiwinner approval voting, fairness is typically understood in terms of proportionality: the elected committee should reflect, to an appropriate extent, the fraction of voters who share similar approval sets. 
This guarantees that sufficiently large groups of voters receive adequate representation, while preventing smaller groups from being disregarded. 
Proportionality axioms, like justified representation (JR) and extended justified representation (EJR), formalize this intuition and have become a central concept in the study of fairness in multiwinner elections. 
Given the size of the voter set $n$ and committee size $k$, a group of voters $S$ of size at least $n/k$ that agrees on a common candidate is called \emph{cohesive} and should be able to demand some form of representation in the committee. 
The Hare quota ($=n/k$) is deeply ingrained in the definitions of most proportionality axioms in multiwinner voting, including JR and EJR, as well as proportional justified representation (PJR) and full justified representation (FJR). 

Nevertheless, there are various situations where these proportionality axioms fail to capture intuitive fairness. 
JR demands that every cohesive group $S$ has at least one voter $v \in S$ who approves a candidate in the committee. 
Consider, for example, the election instance provided in \Cref{fig:jr_drawbacks_simple_example}.
For the group $S$ to be cohesive, it must contain at least $5$ voters. 
It should be easy to see that the committee $W = \{c_3,c_4\}$ satisfies JR, and, thereby, disregards the groups $\{1,2,3,4\}$ and $\{7,8,9,10\}$.
These groups fall just short of meeting the quota.
This instance, however, permits a more balanced representation of all voters by selecting the committee $W' = \{c_1,c_2\}$. 

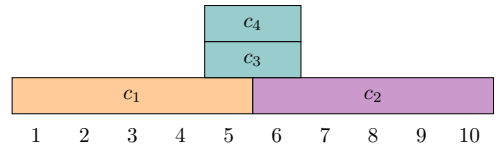
\begin{figure}[!ht]
    \centering
        \scalebox{0.8}{
    \begin{tikzpicture}[yscale=0.6,xscale=0.8, voter/.style={anchor=south}] 
        \foreach \i in {1,...,10}
            \node[voter] at (\i-0.5, -1) {$\i$};
        
        \draw[fill=orange!40] (0, 0) rectangle (5, 1);      
        \draw[fill=violet!40] (5, 0) rectangle (10, 1);     
        \draw[fill=teal!40] (4, 1) rectangle (6, 2);
        \draw[fill=teal!40] (4, 2) rectangle (6, 3);
        
        \node at ( 2.5, 0.5) {$c_{1}$};
        \node at ( 7.5, 0.5) {$c_{2}$};
        \node at ( 5, 1.5) {$c_{3}$};
        \node at ( 5, 2.5) {$c_{4}$};

    \end{tikzpicture}}   
    \caption{
    A voting instance with 10 voters, candidate set $\{c_1, c_2, c_3,c_4\}$ and $k=2$ in which any committee containing $c_3$ satisfies JR. 
    Voters are denoted by an integer and approve the candidates stacked above them.
    }
    \label{fig:jr_drawbacks_simple_example}
\end{figure}

This phenomenon extends to all fairness axioms defined in relation to these static quota requirements. 
Even the core is not immune to this issue.\footnote{For a formal definition, consider the paper by \citet{peters2021market}.}
\Cref{fig:core_drawbacks_simple_example} illustrates an example in which the committee $\{c_1,c_2,c_3\}$ belongs to the core, even though voters $5$–$12$ remain unrepresented.
The three voters supporting $c_4$ once again fall short of the quota and can therefore not form a valid deviating coalition.

\begin{figure}[!ht]
    \centering
        \scalebox{0.8}{
    \begin{tikzpicture}[yscale=0.6,xscale=0.8, voter/.style={anchor=south}] 
        \foreach \i in {1,...,12}
            \node[voter] at (\i-0.5, -1) {$\i$};
        
        \draw[fill=orange!40] (0, 0) rectangle (4, 1);
        \draw[fill=orange!40] (0, 1) rectangle (4, 2); 
        \draw[fill=orange!40] (0, 2) rectangle (4, 3); 
        
        \draw[fill=violet!40] (4, 0) rectangle (7, 1);     
        \draw[fill=teal!40] (7, 0) rectangle (8, 1);

        \draw[fill=violet!40] (8, 0) rectangle (11, 1);     
        \draw[fill=teal!40] (11, 0) rectangle (12, 1);
        
        \node at ( 2, 0.5) {$c_{1}$};
        \node at ( 2, 1.5) {$c_{2}$};
        \node at ( 2, 2.5) {$c_{3}$};

        \node at ( 5.5, 0.5) {$c_{4}$};
        \node at ( 7.5, 0.5) {$c_{5}$};
        \node at ( 9.5, 0.5) {$c_{6}$};
        \node at ( 11.5, 0.5) {$c_{7}$};
        
    \end{tikzpicture}}   
    \caption{A voting instance with 12 voters, candidate set $\{c_1, \dots, c_7\}$ and $k=3$ in which the committee $W = \{c_1, c_2, c_3\}$ lies within the core even though a large group of voters remains unrepresented (example taken from \citep{peters2021market}).}
    \label{fig:core_drawbacks_simple_example}
\end{figure}

These observations regarding the limitations of proportionality are well established in the literature: \\ 
\begin{quote}
    "Properties like the core (and also EJR, PJR, and JR) prevent specific pathological situations, but beyond their definitions, do not provide intuitive justifications for why a committee they allow should be selected" --- \citet{peters2021market}.
\end{quote}

For some proportionality axioms, representation cannot be guaranteed for groups smaller than the Droop quota ($=\frac{n}{k+1}$) \citep{droop1881methods}. 
Hence, fixed quotas may yield examples like that in \Cref{fig:jr_drawbacks_simple_example} or render the axiom unsatisfiable. 
This raises the question of whether, instead of relying on a uniform quota across all instances, one should adapt it with respect to the preference profile.
Our experiments show that the quota values for which JR or EJR remain satisfiable tend to be much smaller than $\frac{n}{k+1}$.
A similar observation has been made for the ordinal case \citep{bardal2025proportional}. 

This paper analyzes variable quotas, contributing to and expanding upon the existing literature \citep{janson2018thresholds, bardal2025proportional}. 
The central idea is to reduce the quota $n/k$ by a factor $\alpha$, approaching the infimum value at which JR or EJR ceases to be feasible. 
In the case of JR, finding the optimal $\alpha$-value corresponds to minimizing the size of the largest group of unrepresented voters who approve a common candidate.

To illustrate this concept more concretely, consider one more example: 
A company has four locations, each comprising five members. 
The task is to form a representative committee of three members. 
Within each location, all members approve a location representative who is aware of the problems at that location. 
Each location has a manager, a salesperson, and an engineer, with each role group approving a representative for their role. 
The remaining two members per location are working students, who only approve their location representative. 
Known approval-based multiwinner rules, such as PAV or CC, elect three of the four location representatives because they receive the broadest approval support. 
However, this outcome fails to represent an entire location. 
We argue that a more proportionally fair outcome — and the optimal committee with respect to $\alpha$-JR — would instead include the three group representatives, ensuring that no large and homogeneous segment of the electorate, such as an entire company location, remains unrepresented. 

\paragraph{Our Contributions}
In this paper, we study a class of proportionality axioms, denoted $\alpha$-JR and $\alpha$-EJR, which generalize the existing proportionality axioms JR and EJR. 
There are instances for which existing voting rules perform very badly with respect to these axioms. 
More specifically, we demonstrate that all commonly studied voting rules can produce committees with an additive gap to the optimal $\alpha$-value of $\frac{k^2}{(k+1)^2}$. 
Additionally, we show that computing the optimal $\alpha$-value for a voting instance is NP-hard, even for JR. 
To find the optimal value for JR, we formulate an ILP. 
Further, in the VI and CI domains, we can efficiently compute the optimal committee with respect to $\alpha$-JR. 
For party-list profiles, we provide an algorithm that constructs an optimal $\alpha$-EJR committee in polynomial time, corresponding to the d'Hondt method in the apportionment setting. 
Lastly, we provide empirical evaluation on the distribution of the optimal $\alpha$-value for synthetically generated instances and show that existing voting rules can, to some extent, approximate the optimal $\alpha$-value on average.

\begin{table*}[!ht]
    \centering
    \renewcommand{\arraystretch}{1.4}
    \resizebox{\textwidth}{!}{%
    \begin{tabular}{c|c|c|c}
        Setting & $\alpha$-JR & $\alpha$-EJR & $\alpha$-EJR+ \\
        \hline
        Computing $\alpha_{\Phi}^*(I)$ 
            & NP-complete (Thm \ref{theorem:JR_hardness})
            & NP-hard (Thm \ref{theorem:JR_hardness})
            & NP-complete (Thm \ref{theorem:JR_hardness})\\

        Computing $\alpha_{\Phi}^W(I)$ 
            & $O(nm)$ (Thm \ref{proposition:JR_complexity_given_committee}) 
            & coNP-hard (Thm \ref{theorem:EJR_hardness_given_committee}) 
            & $O(mnk)$ (Thm \ref{proposition:ejrplus_complexity_with_committee}) \\
        \hline
        \multicolumn{4}{c}{\textbf{Restricted domains (computing $\alpha_{\Phi}^*(I)$)}} \\
        \hline
        Party-list profiles 
            & $O(\lvert P\rvert)$ (Thm \ref{theorem:party_domain_ejr}) 
            & $O(k \log(\lvert P\rvert)  + \lvert P\rvert)$ (Thm \ref{theorem:party_domain_ejr})
            & $O(k \log(\lvert P\rvert)+ \lvert P\rvert)$ (Thm \ref{theorem:party_domain_ejr})\\
        VI domain 
            & $O(n^2 m \log n)$ (Thm \ref{theorem:vi_domain_jr}) & \textbf{?} & \textbf{?} \\
        CI domain 
            & $O(m^2n \log n)$ (Thm \ref{theorem:ci_domain_jr}) & \textbf{?} & \textbf{?} \\
    \end{tabular}
    }
    \caption{Overview of our computational results. 
    Cells with a question mark indicate that the question is still open.}
    \label{table:overview}
\end{table*}

% ---------------------------------------------------- %
\section{Related Work}
Our work builds on the growing literature on fairness and proportionality in approval-based multiwinner voting (see, e.g., \citep{lackner2023multi, ABC+16a, sanchez2017proportional, skowron2017proportional}). 
The central idea in this line of research is that voters who agree on a set of candidates should receive influence in the elected committee proportional to their size. 
The notions of JR and EJR, introduced by \citet{ABC+16a}, translate this into the setting of multiwinner approval voting, thereby marking the starting point of proportionality in approval-based multiwinner voting. 
Building on these axioms, subsequent work has extended proportionality axioms and analyzed their implications for prominent rules \citep{BrPe23a, brill2024phragmen}. 
These ideas have also been generalized beyond multiwinner elections to participatory budgeting \citep{PPS21a}, and general voting models with feasibility constraints \citep{masavrik2024generalised}. 

One aspect that has received relatively little attention is the role of quotas, i.e., the requirement that a group of like-minded voters must reach a certain threshold to secure representation in the committee. 
The most common quota used in the literature is the Hare quota ($= \frac{n}{k}$).
For this, it has been shown that JR and EJR are always satisfiable \citep{ABC+16a}. 
This also holds for some stronger versions of the axioms, like EJR+ and FJR. 
More recently, these representation axioms have been considered for the Droop quota ($=\frac{n}{k+1}$) \citep{casey2025justified}. 
In this paper, the authors demonstrate that proportionality axioms based on the Droop quota can be satisfied for any given voting instance. 
They show that many voting rules and concepts originally developed for the Hare quota can be effectively adapted to the Droop quota setting. 
Before this work, the Droop quota was primarily studied within the framework of apportionment \citep{pukelsheim2017proportional, Brill_2020}.

In this work, we go a step further and consider dynamic quotas. 
\citet{droop1881methods} — after whom the quota is named — already observed that elections often result in situations where more voters than the quota support a common candidate, raising the question of how to allocate the “excess” votes that such candidates receive. 
That being said, we take a quantitative perspective of proportionality, aiming to measure the extent to which a given committee satisfies proportional representation. 
Similar quantitative approaches have been explored in prior work \citep{bardal2025proportional, lackner2019quantitative, janson2018thresholds, janeczko2022complexity}.

The work closest to ours is that of \citet{janson2018thresholds}, who studied and surveyed thresholds of proportionality. 
More specifically, he asked for the smallest proportion of voters that can be guaranteed a good outcome, where a good outcome corresponds to committees satisfying various proportionality notions. 
While he mostly focused on adjusting the proportion of votes achievable by sequential Phragmén and Thiele rules, we consider additional voting rules and examine computational approaches and results related to minimizing the quota. \\
More recently, \citet{bardal2025proportional} developed quantitative measures of proportionality in the ordinal setting, thereby strengthening the proportionality guarantees of so-called \emph{solid coalitions}. 
Their work primarily focused on empirical analyses of real-world data, whereas our study centers on the theoretical aspects of optimizing the $\alpha$-value in the approval setting and on exploring results in the voter- and candidate-interval domains. 
In parallel to us, they study the parameter $\alpha$, which for values less than $1$ makes \emph{proportionality for solid coalitions} more demanding. 
\citet{SLB+17a} formalized proportionality in the ranking setting when voters have approval preferences and introduced a measure of group satisfaction within this setting.

Alternative approaches to capturing fairness in multiwinner voting have been explored beyond JR–based axioms \citep{brill2024completing, kehne2025robustcommitteevotingrepresentation, lackner2020utilitarian}. 
\citet{brill2025individual} introduced the axiom of \emph{individual representation (IR)}, which strengthens the earlier concept of \emph{semi-strong justified representation} proposed by \citet{ABC+16a}. 
While \emph{semi-strong JR} requires that every sufficiently large and cohesive group of voters is represented in the committee at least once, IR refines this principle by assigning each voter a guarantee relative to the largest $\ell$-cohesive group she is able to form for a given profile. 
IR is not always satisfiable: there exist profiles for which no committee provides individual representation, and non-trivial approximation guarantees are unattainable without imposing restrictions on the domain of admissible profiles. 
\citet{kehne2025robustcommitteevotingrepresentation} shifted the focus on fairness to be more candidate-centric. 
Their desideratum of fairness states that candidates with similar sets of supporters should receive similar representation.

Lastly, we mention the \emph{proportionality degree (PD)} as an alternative for quantifying proportionality in multiwinner approval voting. 
This approach strictly enforces the quota while allowing greater flexibility in determining each group's average satisfaction level. 
\citet{skowron2021proportionality} used this to measure the proportionality guarantees that different voting rules can achieve. 
He showed that the sequential variant of the proportional approval voting rule has higher PD than the sequential Phragmén rule. 
\citet{janeczko2022complexity} looked into computational aspects of the proportionality degree in committee elections. 
They show that deciding if a committee with a given PD exists is NP-hard and that verifying if a given committee provides a given PD is coNP-complete.

% ---------------------------------------------------- %
\section{Preliminaries}\label{section:preliminaries} 
We consider the standard approval-based multiwinner voting framework. 
For $x \in \N$, let $[x] = \{1, 2, \dots, x\}$. 
Let $N =\{v_1, \dots, v_n\}$ be the set of voters, $C = \{c_1, \dots, c_m\}$ be the set of candidates, and $k \in \N$ be the committee size with $k \leq m$. 
For every voter $v \in N$, we denote the set of candidates she approves of as  $A_v \subseteq C$, and we write $N_c = \{ v \in N: c \in A_v \}$ for the set of supporters of a candidate $c \in  C$. 
A voting instance $I$ is defined by an approval profile $A = (A_v)_{v \in N}$ and the committee size $k$, i.e., $I = (A, k)$.
Finally, a committee $W$ is a subset of the candidates with cardinality $k$.
To determine which committee to select, various multiwinner voting rules have been proposed. 
In this paper, we consider commonly used approval-based multiwinner voting rules, including the method of equal shares (MES), the sequential Phragmén (seq-Phragmén), the Chamberlin-Courant (CC), proportional approval voting (PAV), and the greedy justified candidate rule (GJCR). 
We refer the reader to the full version of the paper for the formal definitions of these rules.

% ---------------------------------------------------- %
\section{Quota-Based Strengthening of Proportionality} \label{section:quota_definitions}
As stated in the beginning, most proportionality axioms rely on the principle that a group is entitled to representation only if it satisfies some cohesiveness condition. 
Cohesiveness in its general form demands that a voter group satisfy the Hare quota \citep{ABC+16a, sanchez2017proportional}.

\begin{definition}[$\ell$-Cohesiveness \citep{ABC+16a}]
    For $\ell \in [k]$, a group $S \subseteq N$ is $\ell$-cohesive if (1) $\lvert S \rvert \geq \ell \cdot \frac{n}{k}$ and (2) $\lvert \bigcap_{v \in S}A_v \rvert \geq \ell$.
\end{definition}

In this work, we study a more general version of this definition, which incorporates a parameter $\alpha$ to make the quota requirement dynamic. 
The parameter $\alpha$ acts as a scaling factor that adjusts the group size required for representation. 
Smaller $\alpha$ values correspond to a more demanding standard of proportionality, as even smaller groups can claim representation.

\begin{definition}[$(\alpha, \ell)$-Cohesiveness] \label{definition:alpha_ell_cohesiveness}
    For $\ell \in [k]$ and $\alpha \in \R_{\geq0}$, a group $S \subseteq N$ is $(\alpha, \ell)$-cohesive if (1) $\lvert S \rvert \geq \alpha \cdot \ell \cdot  \frac{n}{k}$ and (2) $\lvert \bigcap_{v \in S}A_v \rvert \geq \ell$.
\end{definition}

For $S = \emptyset$, it holds that $\bigcap_{v \in S} A_v = C$.
Changing the quota size has, up until now, mostly been done for $\alpha\geq 1$, which relaxes the cohesiveness definition \citep{DBW+24a, DFPS25a, jiang2020approximately, do2022onlineapproval, halpern2023representation}.

With this generalized cohesiveness definition, we modify existing proportionality axioms:

\begin{definition}[$\alpha$-JR]
    A rule satisfies $\alpha$-JR if for every $(\alpha, 1)$-cohesive group $S$, there exists a voter $v \in S$ with $\lvert A_v \cap W \rvert \geq 1$.
\end{definition}

\begin{definition}[$\alpha$-EJR]
    A rule satisfies $\alpha$-EJR if for every $(\alpha, \ell)$-cohesive group $S$, there exists a voter $v \in S$ with $\lvert A_v \cap W \rvert \geq \ell$.
\end{definition}

\begin{definition}[$\alpha$-EJR+]
    A committee $W$ satisfies $\alpha$-EJR+ if there does not exist a candidate $c \in C \setminus W$, a group of voters $S \subseteq N$, and $\ell \in [k]$ with $\vert S \vert \geq \alpha \cdot \ell \cdot \frac{n}{k}$ such that
    \begin{equation*}
        c \in \bigcap_{v \in S} A_v \text{ and } \vert A_v \cap W \vert < \ell \text{ for all } v \in S.
    \end{equation*}
\end{definition}

Clearly, for the case $\alpha = 1$, the three proportionality axioms correspond to the original notions of JR, EJR, and EJR+. 
Therefore, several voting rules are known to guarantee satisfiability. 
More recently, it has been established that for any given instance, there exists an $\alpha < 1$ such that $\alpha$-EJR+ can be satisfied.

\begin{theorem}[\citet{casey2025justified}] \label{theorem:droop_existence}
    For every instance $I$, and any $\alpha$ with $\alpha > \frac{k}{k+1}$ there exists a committee $W$ that satisfies $\alpha$-EJR+ and is polynomial-time computable.
\end{theorem}

In this paper, we are interested in the minimal achievable quota for a given instance $I$.
Therefore, we would like to find the infimum $\alpha$-value that still allows for an instance to satisfy $\alpha$-$\Phi$ where $\Phi \in \{\text{JR}, \text{EJR}, \text{EJR+}\}$.
That being said, let 
\begin{equation*}
    \alpha_{\Phi}^W(I) := \inf \{ \alpha \colon W \text{ satisfies }\alpha\text{-}\Phi \}
\end{equation*} 
for the given instance $I$.
We refer to $\alpha_{\Phi}^W(I)$ as the $\alpha$-value of a committee $W$ with respect to the proportionality axiom $\Phi$.
The minimum achievable $\alpha$-value for an instance $I$ is denoted by 
\begin{equation*}
    \alpha^*_{\Phi}(I) = \min_{W \subseteq C: \lvert W \rvert = k} \alpha^W_{\Phi}(I).
\end{equation*}

We also refer to this as the optimal $\alpha$-value for an instance $I$ with respect to $\Phi$. 
If the instance $I$ is clear from the context, we omit it and just write $\alpha_{\Phi}^*$ or $\alpha_{\Phi}^W$. 

The notion of $\alpha^*_{\Phi}$ is equivalent to finding the largest $\alpha$-value for which all committees $W$ violate $\alpha$-$\Phi$. 
Observe that by definition, the value is greater than or equal to $0$.
\footnote{$S = \emptyset$ is $(0, \ell)$-cohesive and violates $\alpha$-$\Phi$ for all possible committees $W$.}  

If all voters with a non-empty ballot approve a candidate in the committee, we have $\alpha^*_{\text{JR}}=0$, and if $\lvert A_v \cap W \rvert \geq \min(\lvert A_v \rvert, k)$ for all $v \in N$, then $\alpha^*_{\text{EJR}}=0$.
\footnote{For $\alpha > 0$, every voter group $S$ with $\lvert S \rvert \geq \alpha \ell \frac{n}{k}$ has at least one "satisfied" voter.}
Note that $\alpha \geq 0$ since the proportionality axioms are always violated when $\alpha=0$.
For $\alpha > \frac{k}{k+1}$, all proportionality axioms can be satisfied, which was shown by \citep{casey2025justified}.

\subsection{General Properties}
This approach generates a whole class of proportionality axioms. 
We, therefore, start by showing some general relational results:  
The known implications chain of EJR+, EJR, and JR only holds for $\alpha_1$-EJR+, $\alpha_2$-EJR, and $\alpha_3$-JR when $\alpha_1 \leq \alpha_2 \leq \alpha_3$. 
Note that all proofs can be found in the full version of the paper.
\begin{restatable}{proposition}{propImplications}
\label{prop:implications}
    For $\alpha_1 \leq \alpha_2 \leq \alpha_3$, it holds that 
    $\alpha_1$-EJR+ $\implies \alpha_2$-EJR $\implies \alpha_3$-JR.
\end{restatable}

\begin{proof}[Proof Sketch]
    Let $\alpha \leq \alpha'$ and $\Phi \in \{\text{JR, EJR, EJR+}\}$.
    First, $\alpha$-$\Phi$ implies $\alpha'$-$\Phi$ since as $\alpha$ increases, the set of cohesive groups shrinks. 
    Showing the implications between $\alpha'$-EJR+, $\alpha'$-EJR, and $\alpha'$-JR then follows directly from the original definitions of the proportionality axioms.
\end{proof}

\begin{restatable}{proposition}{EJRJRincompatible}
\label{prop:NoImplications}
    For $1 \geq \alpha_1 > \alpha_2 > \alpha_3$, it holds that $\alpha_1$-EJR+, $\alpha_2$-EJR, and $\alpha_3$-JR are incomparable.
\end{restatable}

We next demonstrate that there can be a substantial gap between the optimal $\alpha$-values for JR and EJR. 
This highlights that our framework offers greater flexibility in how representation is understood. 
Under the classical definitions, it is natural to focus on EJR, as it always implies JR. 
However, introducing the parameter $\alpha$ allows us to adjust the strength of representation: we can choose whether to prioritize smaller, cohesive groups or to emphasize the representation of larger ones.

\begin{restatable}{proposition}{JrEjrApart} \label{prop:additive_jr_ejr_apart}
     There is an instance $I$ such that the additive difference of $\alpha^*_{\text{JR}}(I)$ and $\alpha^*_{\text{EJR}}(I)$ is $k/(k+1)$. 
\end{restatable}

\begin{proof}[Proof Sketch]
    \begin{figure}[!ht]
        \centering
            \scalebox{0.8}{
        \begin{tikzpicture}[yscale=0.6,xscale=0.8, voter/.style={anchor=south}] 
            \node[voter] at (1-0.5, -1) {$1$};
            \node[voter] at (2-0.5, -1) {$2$};
            \node at (3.5, -0.6) {\Large $\dots$};
            \node[voter] at (6-0.5, -1) {$x-1$};
            \node[voter] at (7-0.5, -1) {$x$};
            \node[voter] at (8-0.5, -1) {$x+1$};
            \node at (8.5, -0.6) {\Large $\dots$};
            \node[voter] at (10-0.5, -1) {$x+y$};

            \draw[fill=orange!40] (0, 0) rectangle (7,1);     \node at (3.5, 1.6) {\Large $\vdots$};
            \draw[fill=violet!40] (0, 2) rectangle (7, 3);     
            \draw[fill=teal!40] (7, 0) rectangle (10, 1);
            
            \node at ( 3.5, 0.5) {$c_{1}$};
            \node at ( 3.5, 2.5) {$c_{k}$};
            \node at ( 8.5, 0.5) {$c_{k+1}$};
        \end{tikzpicture}}   
        \caption{A voting instance
        $N = [x+y]$, $C =\{c_1, \dots c_{k+1}\}$ in which the optimal $\alpha$-value for JR is $0$, whereas the optimal $\alpha$-value for EJR is $\min(\frac{x}{x+y},\frac{k\cdot y}{x + y})$. 
        }
        \label{fig:example_jr_incompatible}
    \end{figure}
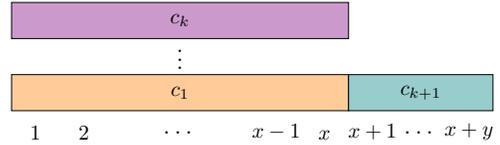
   The committee $W = \{c_1, \dots, c_{k-1}, c_{k+1}\}$ achieves $\alpha_{\text{JR}}^W(I) = 0$, while for $\alpha^*_{\text{EJR}}$ we need to decide whether to select candidates $c_1, \dots, c_k$ or to replace one of them with $c_{k+1}$.
   Either we get a violation by a $1$-cohesive group or by a $k$-cohesive group.
   When we set $x = ky$ in \Cref{fig:example_jr_incompatible}, we can show that $\alpha^*_{\text{EJR}} = k/(k+1)$.
\end{proof}

This construction shows that even when a committee perfectly satisfies $\alpha$-JR ($\alpha=0$), its $\alpha$-EJR value can be as large as $k/(k+1)$, demonstrating a maximal additive gap. 
Since $\alpha$ can be any value in $\mathbb{R}_{\geq0}$, we prove that it is sufficient to only consider $k\cdot n$ many $\alpha$-values to determine the optimum for JR and EJR. 
This technicality is important when designing efficient algorithms and is utilized in the restricted domain section.

\begin{restatable}{lemma}{numberAlphasJR}
\label{lemma:alphasToConsiderJR}
    For a given instance $I$, at most $\lceil \frac{n}{k} \rceil - 1 (\leq n)$ $\alpha$-values need to be checked to determine $\alpha^*_{\text{JR}}(I)$.
    Further, it holds that $\alpha^*_{\text{JR}}(I) \in \Q$.
\end{restatable}

\begin{proof}
    For $\alpha$-JR, it is only relevant what size the smallest $1$-cohesive group has. 
    The group size is between $1$, meaning that a single voter forms a $1$-cohesive group (together with a candidate she approves), and $\frac{n}{k}$, the size of the Hare quota. 
    The relevant $\alpha$-values come from the set $X = \{0,\frac{k}{n}, \frac{2k}{n}, \dots, \frac{rk}{n}\}$ with $r = \lceil \frac{n}{k} \rceil - 1$. 
    Observe that if there does not exist a violation for $\alpha = \frac{k}{n}$ then $\alpha^*_{\text{JR}}$ is $0$ by definition.
    
    Now let $\alpha^*_{\text{JR}}$ be the optimal value.     
    Thus, there exists some $\alpha^*_{\text{JR}}$-JR violation witnessed by some $(N',c)$.  
    Therefore, it holds that $\lvert N' \rvert \geq \alpha^*_{\text{JR}} \cdot \frac{n}{k}$. 
    It even has to be the case that this inequality is tight, since otherwise $\alpha^*_{\text{JR}}$ is clearly not optimal with $(N',c)$ also witnessing a violation for a larger $\alpha$. 
    Thus, $\alpha^*_{\text{JR}} =\lvert N' \rvert \cdot  \frac{k}{n} \in X$ which proves the claim.
\end{proof}

\begin{restatable}{lemma}{numberAlphasEJR}
\label{lemma:alphasToConsider}
    For a given instance $I$, at most $n \cdot k$ many $\alpha$-values are needed to determine $\alpha_{\text{EJR}}^*(I)$. 
    Further, it holds that $\alpha_{\text{EJR}}^*(I) \in \Q$.
\end{restatable}

The last two lemmas allow us to assume that $\alpha \in \Q$ since for any $\alpha \in \R \setminus \Q$ we know from the previous two lemmas that there is some larger $\alpha' \in \Q$.

\subsection{Worst-Case Analysis of Existing Rules} \label{section:analysis_voting_rules}
Many known rules can always guarantee the Droop quota.
Nevertheless, there are instances in which these voting rules perform very poorly with respect to the optimal $\alpha$-value.
Intuitively, this follows from the fact that many rules greedily select candidates approved by many voters and do not consider whether underrepresented voters approve of common candidates.
\Cref{fig:visualisation_example1} visualizes such an instance for $k = 5$. 

\begin{example}  \label{example1}
    Consider the following family of instances. 
    Let $N = (k+1) (k+1)$. 
    Furthermore, let $C = B \cup D$, with $B = \{b_1, \dots, b_{k+1}\}$ such that $N_{b_i} = \{(i-1)\cdot (k+1)+1, \dots, i \cdot (k+1)\}$ for all $i \in [k+1]$ and $D = \{d_1, \dots, d_k\}$ such that $N_{d_i} = \{j \in N \colon (j \mod k+1)  = i\}$ for all $i \in [k]$.
    Then each candidate in $D$ and each candidate in $B$ is approved by $k+1$ voters (see \Cref{fig:visualisation_example1} for a visualization of the case $ k=5$).
\end{example}

\begin{figure}[!ht] 
    \centering
    \resizebox{0.85\columnwidth}{!}{
    \begin{tikzpicture}
        \tikzmath{
        \height = 0.37;
        \width = 0.37;
        \cornerpoint = 0.53;
        \numberVerticesW = 6;
        \numberVerticesWW = 5;
        \numberVerticesV = 6;
        \numberVerticesVV = 5;
        \hyperedgeLength =  2;
        \numberEdges = 2;
        }
        \foreach \v in {1,...,\numberVerticesV} {
            \foreach \w in {1,...,\numberVerticesW} {
                \node (v\v-\w) at (\w,\v) {};
            }
        }
    
        \begin{scope}[fill opacity=0.8]
        \foreach \v in {1,...,\numberVerticesVV} {
        \filldraw[fill=violet!40] ($(v\v-1)+(-\cornerpoint,0)$) 
            to[out=90,in=180] ($(v\v-1) + (0,\height)$)
            to[out=0,in=180] ($(v\v-\numberVerticesW) + (0,\height)$)
            to[out=0,in=90] ($(v\v-\numberVerticesW) + (\cornerpoint,0)$)
            to[out=270,in=0] ($(v\v-\numberVerticesW) + (0,-\height)$)
            to[out=180,in=0] ($(v\v-1) + (0,-\height)$)
            to[out=180,in=270] ($(v\v-1) + (-\cornerpoint,0)$);
            }
    
    \foreach \v in {\numberVerticesV,...,\numberVerticesV}{
        \filldraw[fill=orange!40] ($(v\v-1)+(-\cornerpoint,0)$) 
            to[out=90,in=180] ($(v\v-1) + (0,\height)$)
            to[out=0,in=180] ($(v\v-\numberVerticesW) + (0,\height)$)
            to[out=0,in=90] ($(v\v-\numberVerticesW) + (\cornerpoint,0)$)
            to[out=270,in=0] ($(v\v-\numberVerticesW) + (0,-\height)$)
            to[out=180,in=0] ($(v\v-1) + (0,-\height)$)
            to[out=180,in=270] ($(v\v-1) + (-\cornerpoint,0)$);
            }
            
        \foreach \v in {1,...,\numberVerticesWW} {
            \draw ($(v1-\v)+(0,-\cornerpoint)$) 
            to[out=0,in=270] ($(v1-\v) + (\width,0)$)
            to[out=90,in=270] ($(v\numberVerticesV-\v) + (\width,0)$)
            to[out=90,in=0] ($(v\numberVerticesV-\v) + (0,\cornerpoint)$)
            to[out=180,in=90] ($(v\numberVerticesV-\v) + (-\width,0)$)
            to[out=270,in=90] ($(v1-\v) + (-\width,0)$)
            to[out=270,in=180] ($(v1-\v) + (0,-\cornerpoint)$);
            }
    
        \end{scope}
        \foreach \v in {1,...,\numberVerticesV} {
            \foreach \w in {1,...,\numberVerticesW} {
                 \fill (v\v-\w) circle (0.1);
            }
        }
    
        \foreach \v in {1,...,\numberVerticesWW} {
            \fill (v1-\v) circle (0.1) node [below = 7 mm] {$d_{\v}$};
        }
        \foreach \w in {1,...,\numberVerticesV} {
            \fill (v\w-\numberVerticesW) circle (0.1) node [right = 7 mm ] {$b_{\w}$};
        }
    \end{tikzpicture}
    }
    \caption{Visualization of Example \ref{example1} for $k=5$. 
    Voters are indicated by dots, and hyperedges correspond to the candidates.
    A committee $W$ returned by most voting rules is indicated in violet, and the largest violation for $W$ is shown in orange.
    The committee $D = \{d_1, \dots, d_5\}$ minimizes the $\alpha$-value in this instance.}
    \label{fig:visualisation_example1}
\end{figure}
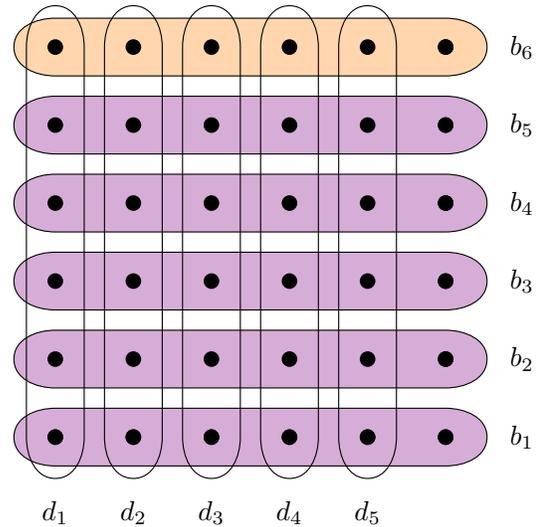

Many existing voting rules achieve an $\alpha$-value of $\frac{k}{k+1}$ in this example, since they return the committee, visualized in violet, while leaving the large group of voters $N_{b_6}$ unrepresented. 
Note that in this instance, rules like GJCR or MES would actually return no candidate, since the instance must contain a JR-violation for them to select any candidate. 
Therefore, we consider a natural way to adapt these rules to $\alpha$ by lowering the quota and giving more budget to the voters, respectively.  
We discuss these adaptations and define all the voting rules mentioned in the extended version. 
In the following, we show that all these rules have a comparably bad $\alpha$-value on the instances defined in \Cref{example1}.

\begin{theorem} \label{theorem:voting_rules_loose}
    For CC, seq-CC, seq-Phragmén, PAV, $\alpha$-MES, and $\alpha$-GJCR, the additive difference to $\alpha^*_{\text{JR}}(I)$ can be of size $\frac{k^2}{(k+1)^2}$.
\end{theorem}

\begin{proof}
    Consider \Cref{example1}
    \footnote{\Cref{example1} can be adapted so that it does not rely on tie-breaking.
    The additive difference for this is then $\frac{k^2}{(k+2)(k+1)}$ (see full version of the paper).}: 
    All of the above rules might choose $k$ from the $k+1$ candidates of the set $B$, i.e., the violet candidates from \Cref{fig:visualisation_example1}.
    Without loss of generality, let $N_{b_{k+1}}$ be the group of voters that remain unrepresented.
    For them not to have a legitimate claim for a candidate in the committee, the quota $\alpha\cdot\frac{n}{k}$ must be greater than $\lvert N_{b_{k+1}} \rvert $.
    Consequently,
    \begin{equation*}
        \alpha > \lvert N_{b_{k+1}}\rvert \cdot \frac{k}{n} = (k+1) \cdot \frac{k}{(k+1)(k+1)} = \frac{k}{k+1}. 
    \end{equation*}
    On the other hand, note that there exists a committee $W = D$, such that for all pairs of non-represented voters $v,v'$ with $v \neq v'$ it holds that $A_v \cap A_{v'} = \emptyset$. Thus, it holds that $\alpha^*_{\text{JR}}(I) \leq \frac{k}{n}.$
    Note that, in this instance, $D$ is an optimal solution since no committee is approved by every voter.
    It follows that the $\alpha$-value of these voting rules can be $\frac{k}{k+1} - \frac{k}{n} =  \frac{k^2}{(k+1)^2}$.
\end{proof} 

For seq-Phragmén, \citet{janson2018thresholds} has already shown that the Droop quota is the smallest achievable quota in every instance.  
We include the result here for completeness. 

\Cref{theorem:voting_rules_loose} complements the results presented by \citet{casey2025justified}.
They demonstrated that the proportionality axioms defined with the Droop quota are still attainable under current voting rules.
We show that these algorithms cannot provide any axiomatic guarantee in the sense of satisfying the more demanding $\alpha$-version of proportional representation notions when $\alpha$ is less than or equal to $\frac{k}{k+1}$ --- corresponding to the Droop quota.

% ---------------------------------------------------- %

\section{Computing the Optimal $\alpha$-value}

In this section, we analyze how to determine the optimal $\alpha$-value for a given instance. 
We begin by showing that, for $\alpha$-JR, potential violations can be detected efficiently for a given committee $W$.

\begin{proposition}\label{proposition:JR_complexity_given_committee}
    For a given instance $I$ and committee $W$, we can compute $\alpha_{\text{JR}}^W(I)$ in $O(nm)$.
\end{proposition}

\begin{proof}
    Given an instance $I$ and a committee $W$, we want to compute the largest $\alpha$ for which there exists an $\alpha$-JR violation. 
    For this, we consider every candidate $c \notin W$ and compute the number of voters that support $c$ but no candidate in $W$. 
    We define the maximal size of such a group as
    \begin{equation*}
        s_{\max} = \max_{c \notin W} \lvert \{v \in N\colon A_v \cap W = \emptyset \text{ and } c \in A_v \}\rvert.
    \end{equation*}
    
    We now claim that the largest $\alpha$ that causes an $\alpha$-JR violation is equal to $s_{\max} \cdot \frac{k}{n}$.
    First, observe that by the definition of $s_{\max}$, there exists a group $S$ of size $s_{\max}$ that commonly approves some candidate $c$ but no candidate in $W$. 
    For $\alpha = s_{\max}\cdot \frac{k}{n}$, $S$ is $(\alpha,1)$-cohesive since 
    \begin{equation*}
        \alpha \cdot \frac{n}{k} = s_{\max}\cdot \frac{k}{n} \cdot \frac{n}{k} = s_{\max},
    \end{equation*} 
    and thus constitutes a violation of $c$. 
    This implies that $\alpha_{\text{JR}}^W\geq s_{\max}\cdot \frac{k}{n}$.
    
    On the other hand, by definition of $\alpha_{\text{JR}}^W$ there exists some $\alpha_{\text{JR}}^W$-JR violation with witness ($N',c)$. 
    Thus, it holds that $\lvert N' \rvert \geq \alpha_{\text{JR}}^W \cdot \frac{n}{k}$. 
    Further, since the voters in $N'$ approve a common candidate $c$ and do not approve any candidate in $W$, it holds that 
    \begin{equation*}
        \lvert N' \rvert  \leq  \lvert \{v \in N\colon A_v \cap W = \emptyset \text{ and } c \in A_v \}\rvert \leq s_{\max}.
    \end{equation*}
    Thus,  $s_{\max} \geq \alpha_{\text{JR}}^W \cdot \frac{n}{k}$ which implies that $\alpha_{\text{JR}}^W \leq s_{\max}\cdot \frac{k}{n}$. \\
    Since we need to compute $s_{\max}$, we must check, for every candidate $c \notin W$, how many voters are not represented and approve $c$. 
    This can be done in $O(nm)$.
\end{proof}

\Cref{proposition:JR_complexity_given_committee} proves that, as for JR, we can determine if there exists an $\alpha$-JR violation for a given committee in polynomial time. 

\begin{restatable}{theorem}{theoremAlphaHardness}\label{theorem:JR_hardness}    
    For a constant $\alpha \in (0, 1)\cap \Q$, deciding whether there exists a committee that satisfies $\alpha$-$\Phi$ with $\Phi \in \{\text{JR, EJR, EJR+}\}$ is NP-hard.
\end{restatable}

NP-completeness of $\alpha$-JR and $\alpha$-EJR+ then follows with \Cref{proposition:JR_complexity_given_committee} and \Cref{proposition:ejrplus_complexity_with_committee}, respectively.
The intuition behind the reduction is as follows: We can ensure that certain candidates are selected if the support they receive is sufficiently large and every voter in that support approves only this specific candidate. 
Further, since $\alpha < 1$, we can choose the support of those candidates small enough that after picking those candidates, there remains a large number of voters unrepresented. 
Avoiding an $\alpha$-JR violation for this group is then as hard as solving a vertex cover problem. 

Since the problem of determining the optimal $\alpha$-value for an instance is computationally hard, we provide a method to determine the value by repeatedly solving the following ILP:\footnote{Note that \Cref{lemma:alphasToConsiderJR} limits the number of $\alpha$-values that need to be considered for the ILP.}

\begin{align*}
    \sum_{c \in C} x_c = k, \quad\quad\quad\quad\;\;  &\hspace{1cm}   (1)\\[6pt]
    y_v \leq \sum_{c \in A_v} x_c \quad \forall v \in N, &\hspace{1cm}   (2)
    \\[6pt]
    \sum_{v :\, c \in A_v} (1 - y_v) \;\leq \; \lceil \alpha \cdot \frac{n}{k} \rceil - 1
    \quad  \forall c \in C, &\hspace{1cm}   (3)\\
    y_v,x_c \in \{0, 1\}. \hspace{1cm}&
\end{align*}

Given some $\alpha$, the ILP finds a committee $W$ satisfying $\alpha$-JR if one exists.  The idea of the ILP is as follows. The binary variables $x_c$ indicate whether the candidate $c$ is included in the committee, and the binary variables $y_v$ can only be $1$ if voter $v$ is covered by the output committee, i.e., $v$ has an approved candidate $c$ in the committee (this means $x_c = 1$). 

Condition (1) ensures that the committee contains precisely $k$ candidates, and condition (2) guarantees that, if $y_v = 1$, the voter $v$ approves at least one candidate in the committee. 
The last condition ensures that at most $\lceil \alpha \cdot n/k\rceil-1$ voters approving some common candidate $c$ are left uncovered, capturing the $\alpha$-JR requirement. 
The constraint is equivalent to saying that at least $\lvert N_c \rvert- \lceil \alpha \cdot n/k\rceil+1$ voters in $N_c$ are covered. 
Note that the ILP does not need an objective function, as we are only interested in feasibility.

\begin{restatable}{theorem}{ILPcorrectness}
    The ILP has a solution if and only if there exists a committee that satisfies $\alpha$-JR.
\end{restatable}

Having addressed the case of JR, we now turn our attention to EJR.
In contrast to JR, verifying or optimizing $\alpha$ for EJR is computationally more demanding. 
In fact, even determining whether a given committee satisfies EJR for $\alpha=1$ is coNP-hard \citep{ABC+16a}. 
This hardness directly extends to determining whether, for a fixed $\alpha$, an $\alpha$-EJR violation exists.

\begin{restatable}{theorem}{thmHardnessExtension} \label{theorem:EJR_hardness_given_committee}
    For a constant $\alpha \in (0,1] \cap \Q$, a  given instance $I$ and committee $W$, determining if $W$ satisfies $\alpha$-EJR is coNP-hard.
\end{restatable}

For a given instance $I$ and committee $W$, figuring out the optimal $\alpha$-value with respect to EJR, i.e. finding $\alpha_{\text{EJR}}^W(I)$, can be solved by considering $\ell$-cohesive groups individually with $1 \leq \ell \leq k$: 

Namely, we can find for every $\ell$ and $W$ the size of the largest group of voters, denoted $s_{\max}(W, \ell)$, such that 
\begin{equation*}
    s_{\max}(W, \ell) = \max_{T \subseteq C, \lvert T \rvert = \ell} \lvert \{ v \in N\colon \lvert A_v \cap W \rvert < \ell \text{ and } T \subseteq A_v \} \rvert.
\end{equation*}

Similar to JR, we can now directly determine the largest $\alpha_{\ell}$ such that there exists an $(\alpha,\ell)$-cohesive group that causes an $\alpha$-EJR violation. 
Now for every $\ell$, we compute $\alpha_{\ell} = s_{\max}(W,\ell) \cdot \frac{k}{\ell \cdot n}$ and take the maximum over these values. 
This gives the largest $\alpha$-value such that there exists an $\alpha$-EJR violation. The reason that this does not give an efficient algorithm lies in the fact that determining $s_{\max}(W,\ell)$ is computationally hard for larger $\ell$. 
This computational complexity does not carry over to $\alpha_{\text{EJR+}}^W(I)$.

\begin{restatable}{proposition}{ejrPlusAlphaFinding} \label{proposition:ejrplus_complexity_with_committee}
  For a given instance $I$ and committee $W$, we can compute $\alpha^W_{\text{EJR+}}(I)$ in $O(mnk)$. 
\end{restatable}

% ---------------------------------------------------- %
\section{Restricted Domains}
Since it is computationally hard to compute $\alpha^*_{\Phi}(I)$ for general instances $I$ and $\Phi \in \{\text{JR, EJR, EJR+}\}$, we look into some restricted domains in which we can compute the optimal $\alpha$-value in polynomial time.
These domains are well-studied restrictions to the general multiwinner voting setting. 
Preferences are aligned along a single dimension, such as voters and candidates positioned on a political spectrum, geographic proximity in local elections, or a preference for candidates differing along a skill or topic dimension. 

\subsection{Party-List Profiles}
A commonly studied restricted domain is that of \emph{party-list profiles}, in which voters’ preferences depend solely on candidates' party affiliations.

\begin{definition}[\citet{lackner2023multi}]
    We say that an approval profile $A$ is a \emph{party-list profile} if for every pair of voters $v_i, v_j \in N$ we have either $A_{v_i} = A_{v_j}$ or $A_{v_i} \cap A_{v_j} = \emptyset$. 
    An election instance $I$ is called a \emph{party-list instance} if (1) $A$ is a party-list profile, and (2) for each voter $v \in N$, it holds that $|A_v| \geq k$.
\end{definition}

For $\alpha$-JR, an optimal committee can be obtained by ordering the parties according to their size and iteratively selecting one candidate from each party in this order until $\lvert W  \rvert = k$. 
If every party has received a candidate, the remaining seats may be filled arbitrarily. 
We can extend this positive result to $\alpha$-EJR.

\begin{restatable}{theorem}{partydomainAlgorithm} 
\label{theorem:party_domain_ejr}
    Computing $\alpha^*_{\text{EJR+}}(I)$ for party-list instances can be done in $O(k \log(\lvert P\rvert) + \lvert P \rvert)$ with $P$ being the set of parties.
    Further, we can compute $\alpha^*_{\text{JR}}(I)$ in $O(\lvert P \rvert)$.
\end{restatable}

In the proof, we use the d'Hondt method from the apportionment setting \citep{d1882systeme} and show that it constructs an optimal committee with regard to $\alpha^*_{\text{EJR+}}(I)$.

\subsection{Voter-Interval Domain}
\begin{definition}[\citet{elkind2015structure}]
    Given an election instance $I$, we say $I$ has \emph{voter-interval (VI)} preferences if there exists a linear order $\sqsubset$ over $N$ such that for all voters $v_1, v_2, v_3 \in N$ and for each candidate $c \in A_{v_1} \cap A_{v_3}$, we have that $v_1 \sqsubset v_2 \sqsubset v_3 \implies c \in A_{v_2}$.
\end{definition}

\citet{ABC+16a} showed that checking whether $W$ provides EJR for an instance $I$ is coNP-complete.
However, to the best of our knowledge, there is no result on efficiently finding violations in $\ell$-cohesive groups when we consider only restricted approval profiles. 
More specifically, we prove the following lemma, which allows us to efficiently check for an $\alpha$-EJR violation given an instance $I$ and a committee $W$.

\begin{restatable}{theorem}{viFindingEJRViolation} \label{lemma:vi_cohesive_group}
    Given an instance $I$ in the VI domain, we can verify whether a committee $W$ satisfies $\alpha$-EJR in $O(n^3 k)$.
\end{restatable}

Furthermore, we show that we can compute the optimal value $\alpha_{\text{JR}}^*(I)$ for a given instance $I$ in the VI domain.
This shows that, by restricting voters' preferences, we can obtain a positive result in contrast to the hardness result shown in \Cref{theorem:JR_hardness} for general instances. 

\begin{restatable}{theorem}{viDomainJR} \label{theorem:vi_domain_jr}
    Computing $\alpha^*_{\text{JR}}(I)$ in the VI domain can be done in $O(n^2 m \log n)$.
\end{restatable}

The greedy algorithm in \Cref{alg:vi_domain} constructs a subset of candidates that satisfies $\alpha$-JR for a given $I$ and $\alpha$, which is cardinal-minimal.\\
By checking if the size of the returned committee is at most $k$, we can verify if, for this $\alpha$-value, we obtain a witness (feasible committee) satisfying $\alpha$-JR.
Due to \Cref{lemma:alphasToConsiderJR}, the number of distinct $\alpha$-values we need to consider is upper-bounded by $n$.
We can go through the $\alpha$-values using binary search to find the highest value such that \Cref{alg:vi_domain} returns a committee that has more than $k$ candidates, as this shows that there does not exist a committee that can satisfy $\alpha$-JR for the $\alpha$-value given in this step.

In \Cref{alg:vi_domain}, we use the observation that in the VI domain, every candidate $c$ has a left-most and right-most supporting voter according to the VI order. 
We denote these voters by $l_c$ and $r_c$, respectively. 
For $i \in [n]$, the algorithm considers the set $\{v_1, \dots v_i\}$ and fixes an $\alpha$-JR violation, if existent, with a candidate approved by the right-most voter.
Intuitively, this ensures that the violation to the left of $v_i$ is resolved while covering as many voters as possible to the right, thereby postponing the next potential violation as long as possible.

\begin{algorithm}
    \caption{Finding Optimal Committees in the VI Domain}
    \label{alg:vi_domain}
\begin{algorithmic}
    \State \textbf{Input:} Election $I = (A, k)$ in the VI domain, $\alpha \in \R_{\geq 0}$
    \State \textbf{Output:} Committee $W \subseteq C$ 
    \State $W \gets \emptyset$
    \For{$i = 1, \dots, n$}
        \State Check whether voters $\{v_1, \dots, v_i\}$ admit an $\alpha$-JR 
        \Statex \hspace{1.4em} violation under $W$ 
        \If{a violation is found}
            \State $W \gets W \cup \{c\}$ with $c \in A_{v_i}$ and
            \Statex \hspace{2.6em} $r_{c'} \sqsubseteq r_{c}$ for all $c' \in A_{v_i}$.
        \EndIf
    \EndFor
    \State \Return $W$
\end{algorithmic}
\end{algorithm}

\subsection{Candidate-Interval Domain}
\begin{definition}[\citet{elkind2015structure}]
    Given an election instance $I$, we say that $I$ has \emph{candidate-interval preferences (CI)} if there exists a linear order $\sqsubset$ over $C$ such that for each voter $v \in N$ and all candidates $a,c \in A_v$, $b \in C$ we have that $a \sqsubset b \sqsubset c \implies b \in A_v$. 
\end{definition}

Intuitively, each voter approves a consistent interval of candidates. It turns out that the positive results for the VI domain also hold in the CI domain. 
More specifically, we show that we can find the optimal $\alpha$-value for a given instance $I$ and committee $W$ with respect to EJR.

\begin{restatable}{theorem}{ciFindingEJRViolation} \label{lemma:ci_cohesive_group}
    Given an instance $I$ in the CI domain, we can verify whether a committee $W$ satisfies $\alpha$-EJR in $O(m^2kn)$.
\end{restatable}

Again, as for the VI domain, it is possible to efficiently compute a committee that satisfies $\alpha$-JR for the optimal value of $\alpha$.

\begin{restatable}{theorem}{ciDomainJR} \label{theorem:ci_domain_jr}
     Computing $\alpha^*_{\text{JR}}(I)$ in the CI domain can be done in $O(m^2n \log n)$.
\end{restatable}

We again have at most $n$ $\alpha$-values to check, and for each $\alpha$-value, we use the greedy algorithm from \Cref{alg:ci_domain} to find a minimal committee satisfying $\alpha$-JR. 
The size of the returned committee then indicates the satisfiability of a given $\alpha$-value. 
Here, we iterate over the candidates in the linear order $\sqsubset$ given by the CI domain, removing candidates that can be omitted without $\alpha$-JR violation. 
Intuitively, the algorithm has to keep a candidate if a set of voters with a relatively "early" interval has a violation, and later candidates cannot make up for this violation. 
This results in the algorithm constructing an optimal committee that pushes the candidates to the right regarding the linear order $\sqsubset$ over the candidates.

\begin{algorithm}
    \caption{Finding Optimal Committees in the CI Domain}
    \label{alg:ci_domain}
\begin{algorithmic}
    \State \textbf{Input:} Election $I = (A, k)$ in the CI domain, parameter $\alpha \in \mathbb{R}_{\geq0}$
    \State \textbf{Output:} Committee $W \subseteq C$
    \State $W \gets C$
    \For{$i = 1, \dots, m$}
        \If{$W\setminus \{c_i\}$ satisfies $\alpha$-JR}
            \State $W \gets W\setminus \{c_i\}$
        \EndIf
    \EndFor
    \State \Return $W$
\end{algorithmic}
\end{algorithm}

For both the CI and VI domains, we leave open the question of whether a polynomial-time algorithm exists to compute the optimal $\alpha$-value for EJR, $\alpha_{\text{EJR}}^*(I)$, for a given instance $I$.
The same holds for EJR+.
The proposed left-to-right sweeping methodology does not work when $\ell>1$.

% ---------------------------------------------------- %
\section{Empirical Evaluation} \label{section:experiments}
\begin{figure*}[t]  
    \centering
    \includegraphics[width=0.95\textwidth]{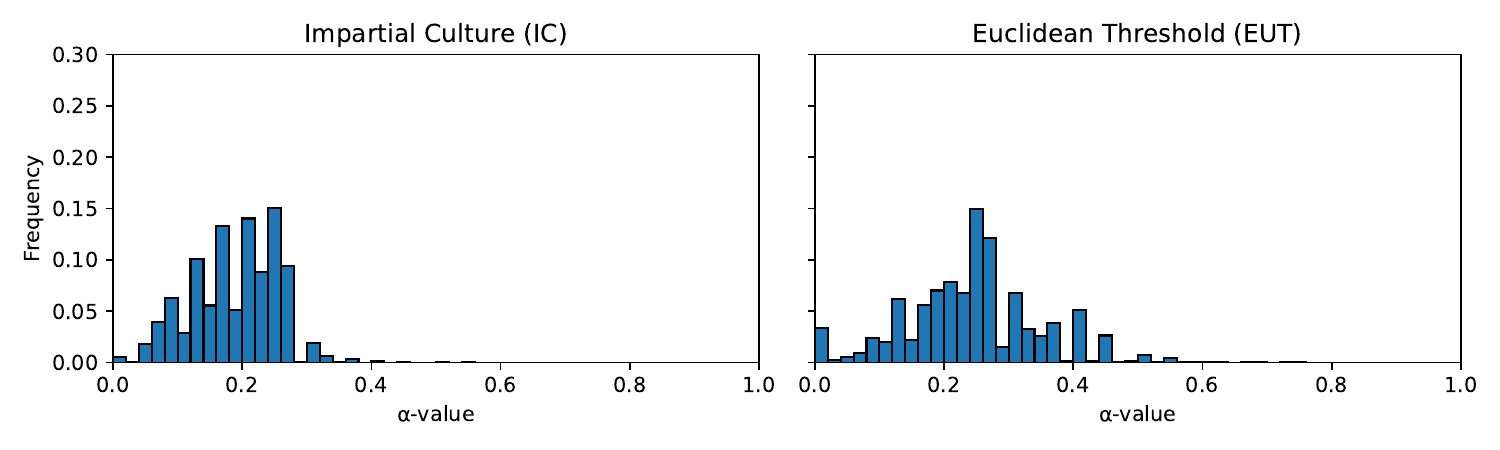}
    \caption{Distribution of optimal $\alpha$-values, $\alpha_{\text{EJR}}^*$, under the IC and the EUT model, based on 6400 generated instances each.}
    \label{fig:alpha_values_overall}
\end{figure*}
\begin{figure*}[t]
    \centering
    \includegraphics[width=0.95\textwidth]{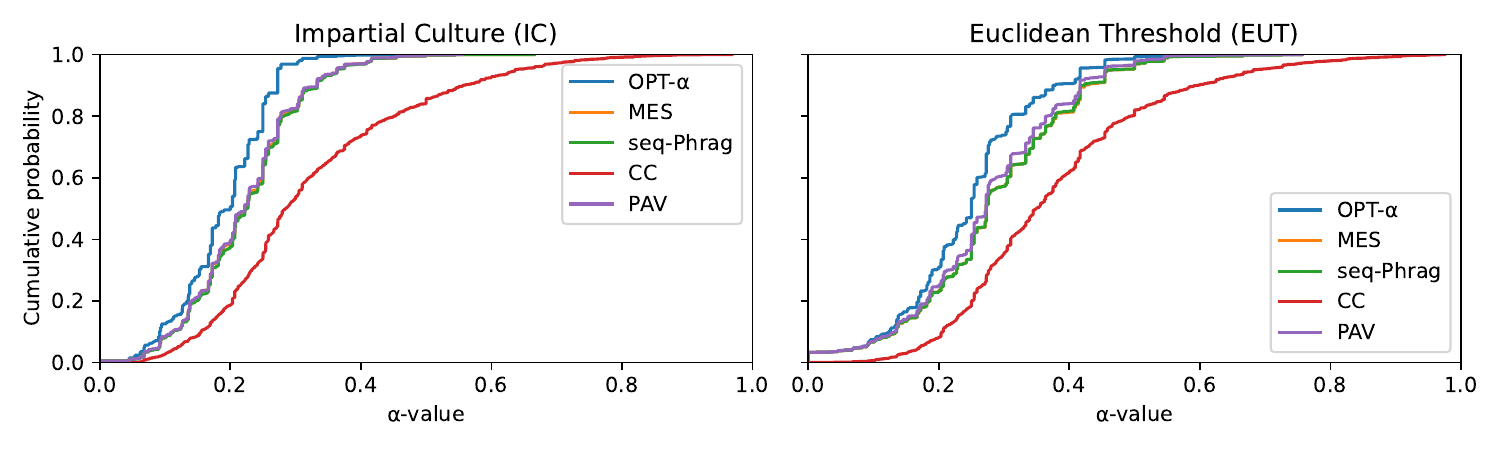}
    \caption{Cumulative distribution of the optimal $\alpha$-value, $\alpha_{\text{EJR}}^*$,  and the $\alpha$-value achieved by MES, seq-Phragmén, CC, and PAV under the IC and the EUT model, based on 6400 samples.}
    \label{fig:ci_cumulative_alpha}
\end{figure*}

Finally, we complement our theoretical results with an empirical analysis. 
We generate synthetic voting instances to evaluate the resulting values of $\alpha$ and to assess the performance of existing multi-winner voting rules. 
Our primary focus is on $\alpha$-EJR; additional plots for $\alpha$-JR are available in the full version. 
\Cref{fig:alpha_values_overall} reports the distribution of optimal values, $\alpha^*_{\text{EJR}}$, under the Impartial Culture (IC) and the Euclidean Threshold (EUT) model, two commonly used probabilistic models for generating approval profiles. 
For each model, we consider instances of varying sizes---ranging from small to moderate---with up to $60$ voters and $15$ candidates. 
Both models are widely used in approval-based multiwinner voting \citep{bredereck2019experimental,faliszewski2017multiwinner, elkind2017multiwinner}. 
Further details on the parameter configurations of them are provided in \Cref{section:appendix_experiments}.

The results demonstrate that in many instances $\alpha$ can be selected to be significantly smaller than $\frac{k}{k+1}$. In fact quite often $\alpha^*_{\text{EJR}}$ is closer to $0$ than it is to $1$. 
\citet{bardal2025proportional} show similar results when analyzing the $\alpha$-value in the ranked setting for \emph{proportionality of solid coalitions}. 

In addition to our theoretical analysis in \Cref{section:analysis_voting_rules}, we evaluate the performance of existing voting rules on randomly generated instances. 
\Cref{fig:ci_cumulative_alpha} demonstrates that MES, seq-Phragmén, and PAV typically achieve committees with comparatively low $\alpha$-values with respect to EJR, indicating that these rules perform well on average. 
Note that MES is completed with seq-Phragmén. 
Our experiments, however, also uncover instances in which all three rules yield committees whose $\alpha$-values exceed the optimal value by more than a factor of four, highlighting substantial worst-case deviations even in average-case data. 
These cases are mostly seen in larger instances. 
In \Cref{section:appendix_experiments}, we provide more fine-grained information on how the voting rules perform under different instance sizes. 
CC does not perform very well, which is unsurprising because its only goal is to satisfy as many voters as possible and because it does not satisfy EJR for all instances.  
In the full version of the paper, we show that, on the contrary, CC performs best out of the considered rules for $\alpha_{\text{JR}}^*$.

\FloatBarrier
% ---------------------------------------------------- %
\section{Conclusion}
We demonstrate that our dynamic extensions of JR and EJR can lead to much fairer outcomes for various instances.
Our axioms possess a clear fairness appeal, as they seek to minimize the underrepresentation of large aligned groups, which in real-world political settings may otherwise foster resentment.
However, as with many other voting rules that aim to maximize an objective, such as CC or PAV, computing the optimal $\alpha$-value is NP-hard, even for JR. 

While existing voting rules can exhibit poor worst-case behavior theoretically, they perform reasonably well on our artificially generated datasets. 
A natural direction for future work is to evaluate these rules on real-world data, for instance, using the Polkadot dataset \citep{boehmer2024approval}. 
When restricting the set of voting instances, we obtain several positive results regarding the efficient computation and verification of $\alpha$-values. 
Empirical results underscore the importance of considering adaptive quotas in multiwinner approval voting and, perhaps, in other settings. 
Further extensions could include generalizing our framework to weakly cohesive groups or developing broader notions of cohesiveness based on justified representation functions. 
Such generalizations may also prove useful in continuous settings, such as budget aggregation. 

\section*{Acknowledgments}
We want to thank our AAMAS 2026 reviewers for their valuable feedback.
Patrick and Fabian are supported by the Deutsche Forschungsgemeinschaft under grant BR 2312/14-1.

%%%%%%%%%%%%%%%%%%%%%%%%%%%%%%%%%%%%%%%%%%%%%%%%%%%%%%%%%%%%%%%%%%%%%%%%
\bibliographystyle{ACM-Reference-Format}

%%%%%%%%%%%%%%%%%%%%%%%%%%%%%%%%%%%%%%%%%%%%%%%%%%%%%%%%%%%%%%%%%%%%%%%%

\clearpage
\appendix

\section{Mentioned Voting Rules} \label{section:voting_rules}
In this section, we formally define the voting rules that we use in the main document.

\paragraph{Chamberlin-Courant Rule (CC):}
Selects all committees that maximize the CC-score, i.e. $\arg\max_{W \subseteq C, \lvert W \rvert = k}\text{CC}(W)$. The CC-score equals the number of voters who approve at least one candidate in the committee. Formally,
\begin{equation*}
    \text{CC}(W) = \lvert \{ v \in N : A_v \cap W \neq \emptyset \} \rvert.
\end{equation*}
Its sequential variant chooses candidates one at a time, maximizing the number of newly covered voters at each step.

\paragraph{Proportional Approval Voting (PAV):}
Select all committees that maximize the PAV-score, i.e. $\arg\max_{W \subseteq C, \lvert W \rvert = k}\text{PAV}(W)$. The PAV score is defined as
\begin{equation*}
    \text{PAV}(W) = \sum_{v \in N}H(\lvert A_v \cap W \rvert),
\end{equation*}
where $H(t) = \sum_{j=1}^t \frac{1}{j}$, is the $t$-th harmonic number, with $H(0) = 0$.

\paragraph{Sequential Phragmén (seq-Phragmén) \citep{brill2024phragmen}:}
The rule \emph{seq-Phragmén} starts with an empty committee and iteratively adds candidates,  always choosing the candidate that minimizes the (new) maximum voter load (under the assumption that previously assigned loads cannot be redistributed).
Let $\bar{x}^{(j)}$ denote the voter loads after round $j$.
At first, all voters have a load of $0$, i.e., $\bar{x}^{(0)}_i = 0$ for all $i \in N$. 
In each round, we keep the already assigned loads, but we may further increase them and give the additional load to a new candidate $c$.
In other words, we require
\begin{equation*}
    \bar{x}^{(j)}_i \geq \bar{x}^{(j-1)}_i \quad \text{for all } i,
\end{equation*}
with equality unless $i \in N_c$.
Moreover, the sum of the loads added in the round should be $1$.
We select the candidate $c$ and the loads $\bar{x}^{(j)}$ that satisfy these conditions and minimize $\max_i \bar{x}^{(j)}_i$.

\paragraph{Method of Equal Shares(MES) \citep{PeSk20a}:}
This rule starts with a budget of $k/n$ per voter and an empty committee. Adding a candidate to the committee incurs a cost of $1$. At each step, MES selects a candidate that minimizes the maximum budget any voter has to pay. A detailed description can be found in Algorithm \ref{alg:mes}. 
Each voter's budget is $b = k/n$.

\begin{algorithm}
    \caption{MES(b)}
    \label{alg:mes}
\begin{algorithmic}
    \State $W \gets \emptyset$
    \State $b(v) \gets b$ for all $v \in N$
    \State $D \gets \{c \in C \setminus W \vert \sum_{v \in N_c} b(v) \geq 1\}$
    \While{$D \neq \emptyset$}
        \For{$c \in D$}
            \State $q(c) \gets \min \{q \geq 0 \vert \sum_{v \in N_c} \min\{b(v), q\} \geq 1\}$
        \EndFor
        \State $c^* \gets \arg\min_{c \in D} q(c)$
        \State $W \gets W \cup \{c^*\}$
        \For{$v \in N_{c^*}$}
            \State $b(v) \gets b(v) - \min \{b(v), q(c^*)\}$
        \EndFor
        \State $D \gets \{c \in C \setminus W \vert \sum_{v \in N_c} b(v) \geq 1\}$
    \EndWhile
    \State return $W$
\end{algorithmic}
\end{algorithm}

Based on the idea by \citet{casey2025justified}, we define $\alpha$-MES to be the algorithm as above, but voters obtain the minimum amount of budget $b$ such that at least $k$ candidates are bought, i.e., $b = \min\{b \in \R_+: \rvert \text{MES}(b)\lvert\geq k\}$.
If more than $k$ candidates can be bought, we stop after $k$ candidates have been bought.

\paragraph{Greedy Justified Candidate Rule (GJCR) \citep{BrPe23a}:}
This rule iteratively detects the largest current EJR+ violation and adds the responsible candidate to the committee.
Under GJCR, adding a candidate incurs a cost of $ 1$, which is uniformly split among all voters involved in the violation. 
For a formal definition, see Algorithm \ref{alg:gjcr_w_prices}. 
    
\begin{algorithm}[!ht]
    \caption{GJCR}
    \label{alg:gjcr_w_prices}
\begin{algorithmic}
    \State $W \gets \emptyset$
    \State $p(v) \gets 0$ for all $v \in N$
    \For{$\ell \in \{k, k-1, ..., 1\}$}
        \State $N(c) \gets \{v \in N_c : \vert A_v \cap W\vert < \ell \}$ for $c \notin W$
        \While{there is $c \notin W$ s.t. $\vert N(c)\vert \geq \ell \cdot \frac{n}{k}$}
            \State Add candidate $c$ maximizing $\vert N(c)\vert$ to $W$
            \State $p(v) \gets p(v) + \frac{1}{\vert N(c) \vert}$ for all $v \in N(c)$
        \EndWhile
    \EndFor
    \State return $W$
\end{algorithmic}
\end{algorithm}
Similar to Droop-GJCR defined in \citet{casey2025justified},
we define $\alpha$-GJCR to be the rule that looks for candidates $c \notin W$ such that $\vert N(c)\vert \geq \alpha\cdot\ell\cdot\frac{n}{k}$.
We minimize the $\alpha$-value until the algorithm chooses at least $k$ candidates.
We again stop after $k$ candidates have been chosen.

\clearpage

\section{Omitted Proofs}
In this section, we present the proofs missing
from the main body of our paper.

\subsection{General Properties}
\label{section:appendixGeneralProp}

\propImplications*
\begin{proof}
    Let $\alpha_1 \leq \alpha_2$ and $\Phi \in \{\text{JR, EJR, EJR+}\}$.
    First, it holds that $\alpha_1$-$\Phi$ $\implies \alpha_2$-$\Phi$ since increasing $\alpha$ can only decrease the number of cohesive groups, which therefore can only make the axiom weaker. 
    Put differently, every $\alpha_2$-cohesive group is also an $\alpha_1$-cohesive group.
    Therefore, $\alpha_1$-EJR+ $\implies$ $\alpha_2$-EJR+.
    Let $\alpha_2$-EJR+ be satisfied by committee $W$.
    Assume for contradiction that $\alpha_2$-EJR is violated with witness $(N', C', \ell)$ where $N' \subseteq N$ and $C'\subseteq C$.
    For all $v \in N'$, $\lvert A_v \cap W \rvert < \ell$ and due to cohesiveness of $N'$, there is a candidate $c \in C'\setminus W$.
    Since $\lvert N_c \rvert \geq \alpha_2 \cdot \ell \frac{n}{k}$, this candidate causes an $\alpha_2$-EJR+ violation; a contradiction.\\
    Next, we show that $\alpha_2$-EJR implies $\alpha_3$-JR.
    Again, $\alpha_2$-EJR implies $\alpha_3$-EJR, and since $\alpha_3$-JR is a special case of $\alpha_3$-EJR, the implication follows.
\end{proof}

\EJRJRincompatible*
\begin{proof} 
    First, we show that $\alpha_1$-EJR+ does not imply $\alpha_3$-JR.
    Select $a,b \in \N$ such that $\alpha_1 > \frac{a}{b} \geq \alpha_3$ and
    consider the following instance. 
    Let $\alpha_3 = \frac{a}{b}$ for $a, b \in \N$, $N = \{v_1, \dots v_b\} $ and $C = \{c_1, c_2\}$, $k = 1$, with $A_v = \{c_1\}$ for $v  \in \{v_1, \dots, v_a\}$. 
    Then, the committee $W = \{c_2\}$ satisfies $\alpha_1$-EJR+. 
    This follows from the fact that since
    \begin{equation*}
        \alpha_1 \cdot \frac{n}{k} > 
        \frac{a}{b} \cdot \frac{n}{k} = \frac{a}{b} \cdot b = a,
    \end{equation*}
    there does not exist any cohesive groups for $\alpha_3$ in the instance. 
    On the other hand, the voters $\{v_1, \dots, v_a\}$ together with candidate $c_1$ are witnessing an $\alpha_3$-JR violation.

    The same construction can also be used to show that $\alpha_1$-EJR+ does not imply $\alpha_2$-EJR and that $\alpha_2$-EJR does not imply $\alpha_3$-JR.
    \\

    Next, we show that $\alpha_2$-EJR does not imply $\alpha_1$-EJR+. 
    Let $\frac{a_1}{b_1} >\alpha_1 > \alpha_2  > \frac{a_2}{b_2}$ for some $a_1,a_2,b_1,b_2 \in \N$.
    We create the instance $N = [ (2\cdot b_1 \cdot b_2) \cdot ( a_1 \cdot b_2 )]$, with $C = \{c_1, \dots, c_{2 \cdot a_1 \cdot b_2}\} \cup \{c\}$ and $k = 2 \cdot a_1 \cdot b_2$ and $A_{v_i} = \{c, c_i\}$ for $i \in [2 \cdot a_1 \cdot  b_2]$. Thus, $\alpha_1 \cdot \frac{n}{k} < \frac{a_1}{b_1} \cdot \frac{(2\cdot b_1 \cdot b_2) \cdot (a_1 \cdot b_2 )}{2 \cdot a_1 \cdot b_2} =  a_1  \cdot b_2$ and 
    $\alpha_2 \cdot \frac{n}{k} > \frac{a_2}{b_2} \cdot \frac{(2\cdot b_1 \cdot b_2) \cdot (a_1 \cdot b_2 )}{2 \cdot a_1 \cdot b_2} =  a_2  \cdot b_1$.    

    We claim that the committee $W = \{c_1, \dots, c_{2 \cdot a_1 \cdot b_2}\}$ satisfies  $\alpha_2$-EJR but violates $\alpha_1$-EJR+. First, observe that the committee violates $\alpha_1$-EJR+.
    This follows, since $\lvert N_c \rvert =
    2 \cdot a_1\cdot b_2
    > 2 \cdot \alpha_1\cdot \frac{n}{k}$, no voter in $N$ approves $2$ candidates in $W$, and $c \not \in W$. On the other hand, we claim that there does not exist an $(\alpha_2, 2)$-cohesive group. Since $\lvert A_v \cap A_{v'}\rvert \leq 1$ for any $v \neq v' \in N$ an $(\alpha_2, 2)$-cohesive group can contain at most one voter. But as computed above, an $(\alpha_2, 2)$-cohesive group has to be of size at least $2\cdot \alpha_2 \cdot \frac{n}{k} > 2\cdot a_2 \cdot b_1 \geq 2$. On the other hand, every voter who approves of any candidate approves of at least one candidate in $W$. Thus, $W$ satisfies $\alpha_2$-EJR. 

    The same example can be used to prove that 
    $\alpha_3$-JR does not imply $\alpha_1$-EJR+.
    
    Lastly, we show that $\alpha_3$-JR does not imply $\alpha_2$-EJR.
    Consider the following instance $N = [4], C = \{c_1, c_2,c_3\} ,k=2$ with $A_v =  \{c_1,c_2\}$ for $v \in [4]$. Then, $W = \{c_1,c_3\}$ satisfies $\alpha_3$-JR for all $\alpha_3 > 0$ since every voter approves at least one candidate in $W$. On the other hand, $W$ does not satisfy $\alpha_2$-EJR for any $\alpha_2 \leq 1$ since $N$ together with $\{c_1,c_2\}$ witness an $(\alpha_2,2)$-EJR violation.
\end{proof}

\JrEjrApart*
\begin{proof}
    Consider a generalization of the example from \Cref{fig:example_jr_incompatible}.
    We have a set of voters $\{1, \dots, x\}$ approving candidates $c_1$ to $c_k$ and a set of voters $\{x+1, \dots, x+y\}$ approving candidate $c_{k+1}$.
    Without loss of generality, let $c_1$ to $c_{k-1}$ be in the committee $W$.
    If $c_{k+1}$ is also in the committee, then the voters approving candidate $c_k$ are ($\alpha_1$, k)-cohesive for an $\alpha_1$-value of $\frac{x}{x+y}$ and, therefore, witness an $\alpha_1$-EJR violation.
    If $c_k$ instead of $c_{k+1}$ is in the committee, then the voters approving $c_{k+1}$ are $(\alpha_2$, 1)-cohesive for an $\alpha_2$-value of $\frac{k\cdot y}{x+y}$, and witness an $\alpha_2$-EJR violation.
    So the optimal $\alpha$-value with respect to EJR is
    \begin{equation*}
        \min(\frac{x}{x+y}, \frac{k \cdot y}{x+y}).
    \end{equation*}
    By setting $x = k \cdot y$, both values simplify to $k/(k+1)$.
    On the other hand, the optimal committee with respect to $\alpha$-JR, $W = \{c_1, c_{k-1} \}$ results in an optimal $\alpha$-value of $0$ as we can cover every voter by $W$.
    Therefore, the additive gap between JR and EJR for the optimal $\alpha$-value is $k/(k+1)$.
\end{proof}

\numberAlphasEJR*
\begin{proof}
    Similar to \Cref{lemma:alphasToConsider}, it also holds for $\alpha$-EJR that there are only a polynomial number of $\alpha$-values that need to be considered. 
    
    For the proof, consider $\ell$-cohesive groups separately for $\ell \in \{1, \dots, k\}$.
    
    We claim that for $\ell$-cohesive groups it is sufficient to consider the $\alpha$-values $X_{\ell}=\{0,\frac{k}{\ell\cdot n},\frac{2k}{\ell\cdot n}, \dots, \frac{r_\ell\cdot k}{\ell\cdot n}\}$ with $r_\ell = \lceil \frac{n \cdot \ell}{k} \rceil-1$. 
    Since for each $\ell$ there are at most $n$ different $\alpha$-values in total, there are at most $k \cdot n$ many such values. Let $X = \cup_{\ell \in \{1, \dots, k\}}X_\ell$ be the set of all $\alpha$-values ordered from small to large.
    
    We now claim that $\alpha_{\text{EJR}}^* \in X$.
    So let $\alpha = \alpha_{\text{EJR}}^*$ be the optimal value.
    Thus, there exists some $\alpha$-EJR violation witnessed by some $(N',C', \ell)$.
    Let this again be the largest violation, meaning that the ratio of $\lvert N' \rvert$ and $\ell$ is maximal.
    Therefore, it holds that $\lvert N' \rvert \geq \alpha \cdot \ell \cdot  \frac{n}{k}$. 
    As for JR, it must hold that this inequality is tight, since otherwise $\alpha$ is clearly not optimal. 
    Thus, $\alpha = \lvert N' \rvert \cdot \frac{k}{\ell \cdot n} \in X_\ell \subseteq X$ which proves the claim.    
\end{proof}

In the main body of the paper, \Cref{example1} presents an instance $I$ that relies on tie-breaking to show that many rules can return suboptimal committees regarding $\alpha^*(I)$.
In the following, we provide an example that does not rely on tie-breaking. 
The argument works the same way as the one for \Cref{example1}, with the difference that the optimal value in this instance is $2 \cdot \frac{k}{n}$ achieved by the committee $W = D$.

\begin{example}  \label{example1Appendix}
    Consider the following family of instances. 
    Let $N = [(k+1)(k+2)]$. 
    Furthermore, let $C = B \cup D$, with $B = \{b_1, \dots, b_{k+1}\}$ such that $N_{b_i} = \{(i-1)\cdot (k+2)+1, \dots, i \cdot (k+2)\}$ and $D = \{d_1, \dots, d_k\}$ such that $N_{d_i} = \{j \in N \colon (j \mod (k+1))  = i\}$ for all $i \in \{1, \dots, k \}$.
    Then each candidate in $D$ is approved by $k+2$ voters and each candidate in $B$ is approved by $k+1$ voters. Since $N_{b_i} \cap N_{b_j} = \emptyset$ for $i \neq j$ most rules would select $k$ candidates from $B$.
\end{example}

\subsection{Complexity Results} \label{sec:omitted_complexity}
\theoremAlphaHardness*
\begin{proof}
    We first prove the statement for JR.
    For hardness, we give a reduction from vertex cover, which is one of Karp's original 21 NP-complete problems \citep{Karp72a}:
    Given a graph $G = (V, E)$, where $V = \{v_1, \dots v_n\}$ is the set of vertices and $E$ the set of edges, the question is whether there exists a vertex cover of size at most $r$. 
    
    Let $\alpha = \frac{a}{b}$, where $a \geq 2$. (If $a = 1$ just multiply $a,b$ with a factor of $2$.)
    
    We construct the following multiwinner voting instance:
    Let $N = V \cup X \cup Y \cup Z \cup F$ with 
    \begin{align*}
        X &= \{x_1, \dots, x_{a (\lvert V \rvert + a-2)}\}, \\
        Y &= \{y_1, \dots ,y_{ar}\},               \\
        Z &= \{z_1, \dots, z_{(b-a-1)(\lvert V \rvert + a-2) + r \cdot (b-a)}\},  \\
        F &= \{f_1, \dots, f_{a-2}\}.  \\
    \end{align*}
    The set of candidates is defined as $C = C_X \cup C_Y \cup C_E$ with 
    \begin{align*}
        C_X &= \{c^X_1, \dots c^X_{\lvert V \rvert + a-2} \}, \\
        C_Y &= \{c^v_{i}\colon i \in [r], v \in V \}, \\
        C_E &= \{c^E_e \colon e \in E \},  \hspace{3cm}
    \end{align*}
    and the committee size $k=\lvert V \rvert + a - 2 + r $. 
    Finally, the approval ballots for all voters are as follows:
    \begin{align*}
        A_v &= \{c^E_e \colon v \in e\} \cup \{c_j^v \colon j \in [r]\} \text{ for all } v \in V, \\
        A_{x_i} &= \{c^X_{ \lceil i / a \rceil}\} \text{ for all } x_i \in X,\\
        A_{y_i} &= \{c^v_{ \lceil i / a \rceil} \colon v\in V \} \text{ for all } y_i \in Y, \\
        A_f &= C_E \text{ for all } f \in F. \\
    \end{align*}

    In the following, we want to argue that there exists a committee that satisfies $\alpha$-JR if and only if there exists a vertex cover of size $r$. For this, we make the following observations about the instance:

    \noindent\textbf{1. Observation:} The voters in $Z$ do not approve any candidate and therefore cannot cause JR violations. \\
    \noindent\textbf{2. Observation:} In total, there are $a (\lvert V\rvert +a-2) + ar + (b-a-1)(\lvert V \rvert + a - 2) + r (b-a) + a-2 + \lvert V \rvert = b(\lvert V\rvert +a -2 + r)$ many voters. 
    Thus, in order to be an $\alpha$-cohesive group, the group has to contain at least $ \alpha \cdot \frac{n}{k} = \frac{a}{b} \cdot \frac{b(\lvert V\rvert +a  -2 + r)}{\lvert V\rvert +a-2 + r} = a$ many voters.

    \noindent\textbf{3. Observation:} Since for every candidate in $C_X$ there exist $a$ voters in $X$ that precisely approve that candidate, this implies that every committee $W$ that satisfies $\alpha$-JR has to contain all candidates $C_X$.

    \noindent\textbf{4. Observation:} Every committee $W$ satisfying $\alpha$-JR has to contain a candidate $c_i^v$ for every $i \in [r]$, otherwise there exists a group of $a$ voters in $Y$ that does not approve of any candidate in $W$. \\

    From the latter two observations, we get that every committee that satisfies $\alpha$-JR only contains candidates in $C_X \cup C_Y$ and that at least $\lvert V \rvert + a-2 + r \;(=k)$ many candidates need to be selected such that no $\alpha$-JR violation can occur. 
    
    Put differently, we need to select one candidate $c_i^v$ for every $i \in [r]$, as well as every candidate in $C_X$. 
    For every such committee, it holds that every voter in $X$ and $Y$ approves at least one candidate and, thus, cannot be part of an $\alpha$-JR violation. 
    
    Lastly, due to our approval ballot construction, every voter in $F$ does not approve of any candidate in such a committee, and at most $r$ voters in $V$ approve a candidate in such committees. \\
    
    With these observations, we can finally prove the equivalence: 
    First, we show that if there exists a vertex cover for the original instance, then there exists a committee that satisfies $\alpha$-JR.
    So let $\{v^{\text{VC}}_1, \dots v_r^{\text{VC}}\}$ be a vertex cover. 
    We create the following committee 
    \begin{equation*}
        W = C_X \cup \{c_{i}^{v_i^{\text{VC}}} \colon i \in [r] \}.
    \end{equation*}
    Observe that every voter of the vertex cover approves at least one candidate in $W$. 
    We claim that this committee satisfies $\alpha$-JR. 
    The only violations can be due to a group of voters in $V \cup F$ and a candidate $c \in C_E$. 
    Since every candidate in $C_E$ is approved by precisely $\alpha \cdot \frac{n}{k} = a$ many voters, a violation would imply that the whole support of a candidate is not represented in $W$. 
    Since the support of every candidate $c \in C_E$ contains two voters  $v_1, v_2 \in V$ such that $\{v_1,v_2\} \in E$, at least one of them approves a candidate in $W$, and thus we have found a committee that satisfies $\alpha$-JR. 

    For the reverse direction, assume that there exists a committee $W$ that satisfies $\alpha$-JR. 
    As argued above, it has to hold that $C_X \subseteq W \subseteq C_X \cup C_Y$. 
    Further, it holds that at most $r$ voters $V'\subseteq V$ approve any candidate in $W$ (at most one for every $c \in W \cap C_y$). 
    We claim that $V'$ forms a vertex cover. 
    Assume for the sake of contradiction that this does not hold. Then, there exists an edge $e = \{u,v\}$ that is not covered by any vertex in $V'$. 
    Thus, neither $u$ nor $v$ approves any candidate in $W$, but both approve the common candidate $c_{e}^{E}$. 
    Further, every voter in $F$ does not approve any candidate in $W$ and also approves $c_{e}^{E}$. 
    Together they form an $\alpha$-JR violation contradicting the assumption that $W$ satisfies $\alpha$-JR.   
    This concludes the proof for $\alpha$-JR. \\
    
    In the following, we argue that the same reduction can also be used to show hardness for $\alpha$-EJR and $\alpha$-EJR+. 
    To this end, we analyze the number of voters who approve of each candidate. 
    It holds that $\lvert N_c \rvert = a$ for every $c \in C_X$ since $a$ voters in $X$ approve the candidate. 
    For every candidate in $C_E$, all $a-2$ voters in $E$ and two voters in $V$ approve the candidate, so again, a total of $a$. 
    For every candidate in $C_y$, $a$ voters in $Y$, and $1$ voter in $V$ approve the candidate for a total of $a+1$ voters. Every 2-cohesive group must be of size at least 
    \begin{equation*}
        \alpha \cdot 2 \cdot \frac{n}{k} = \frac{a}{b}\cdot 2\cdot \frac{b(\lvert V\rvert +a -2 + r)}{\lvert V\rvert +a -2 + r} = 2a.
    \end{equation*} 
    Therefore, it has to hold that $2a \leq a+1$, which only holds for $a = 1$. 
    By assumption, $a > 1$. 
    Thus, the instance does only contain $1$-cohesive groups. Since for $1$-cohesive groups EJR+, EJR, and JR are equivalent, the reduction also works for EJR and EJR+.
\end{proof}

\ILPcorrectness*
\begin{proof}
    First, assume that the ILP has a solution.
    Let this be $x \in \{0, 1\}^m$ and $y \in \{0, 1\}^n$.
    Define $W = \{c \in C\colon x_c =1\}$ as our committee. 
    We claim that $W$ does not have any $\alpha$-JR violation. Assume for the sake of contradiction, there exists one with witness $(N',c')$. 
    This means all voters in $N'$ approve $c'$ and all voters in $N'$ do not approve any candidate in $W$. 
    Therefore, we can conclude from (1) that $x_c = 0$ for all $c \in A_v$ with $v \in N'$. 
    Therefore, using (2), we get that $y_v= 0$ for all $v \in N'$.  
    Finally, considering (3) for candidate $c'$, we get that $\sum_{v \colon\, c' \in A_v} (1-y_v) \leq \lvert N' \rvert$ and that $\lvert N'\rvert \geq \lceil \alpha \cdot \frac{n}{k} \rceil$, which contradicts that $x$ and $y$ being a solution for the ILP.
    
    For the other direction, assume there is a committee $W$ that, for a given $\alpha$, satisfies $\alpha$-JR. 
    We can define the following solution for the ILP. 
    We set $x_c = 1$ if and only if $c \in W$ and $y_v= 1 $ if and only if $W \cap A_v \neq \emptyset$ for all $c \in C, v \in N$. 
    We claim that this is a valid solution for the ILP. 
    First note that since $\lvert W\rvert = k$, condition (1) is satisfied. 
    Consider now condition (2).
    If $y_v = 0$, then the condition trivially holds, so assume that $y_v = 1$. 
    If $y_v = 1$, then by definition there exists a candidate $c \in A_v \cap W$ and therefore $x_c = 1$. 
    So, finally, assume for the sake of contradiction that (3) is violated for some $c$. 
    Then, there exists a group $N'$ such that $\lvert N'\rvert \geq \lceil \alpha \cdot \frac{n}{k} \rceil$ and for every member $v \in N'$, we have $y_v = 0$ and $c \in A_v$. 
    If $y_v=0$, this implies that $A_v \cap W = \emptyset$. 
    But, then $(N',c)$ are witnessing an $\alpha$-JR violation contradicting our assumption that $W$ satisfies $\alpha$-JR. This concludes the proof.
\end{proof}

\thmHardnessExtension*
\begin{proof}
    We give a reduction that applies to essentially all hardness results and counterexamples that already hold for some larger $\alpha$ (in this case, $\alpha = 1$).
    
    We know that for smaller $\alpha$, all conditions get even harder to satisfy. 
    To use the same arguments as in the original instance ($\alpha = 1$), we can rescale accordingly.
    
    Formally, let $\alpha = \frac{a}{b}$ with $a, b \in \mathbb{N}$ (possible by \Cref{lemma:alphasToConsider}). 
    Given any instance $(N, C,k)$ for the Hare quota, we transform the instance to the following instance $(N_{\alpha}, C,k)$ such that 
    $N_{\alpha} = \bigcup_{i = 1}^a N^{(i)} \cup Z$, where each $N^{(i)}$ is an exact copy of the voters in $N$ meaning that the voters approve precisely the same candidates in $C$, and $Z = \{z_1, \dots z_{(b-a) \cdot \lvert N \rvert}\}$ are voters that do not approve any candidate in $C$.
    
    We now claim that in this adapted instance, there exists an $\alpha$-EJR violation for a committee $W$ if and only if in the original instance there was an EJR violation.
    
    So, assume first that $W$ violates EJR. 
    Thus, there exists a witness set $(N',C', \ell)$ with  $\lvert N' \rvert \geq \ell \cdot \frac{n}{k}$, $\lvert C' \rvert \geq \ell$ such that for all voters in $v_i \in N'$ it holds that $\lvert A_{v_i}  \cap W \rvert < \ell$. 
    We now consider the set $N'_{\alpha} =\bigcup_{i}^a N'^{(i)}$ in the transformed instance. 
    First, since $\lvert N_\alpha\rvert  = n_{\alpha} =  b \cdot \lvert N \rvert$, it holds that 
    \begin{equation*}
        \lvert N'_{\alpha} \rvert = a \cdot \lvert N' \rvert \geq a \cdot \ell \cdot \frac{n}{k} = a \cdot \ell \cdot \frac{n_\alpha}{b \cdot k} = \alpha \cdot \ell \cdot \frac{n_\alpha}{k}.
    \end{equation*}
    By definition, every voter in $ N'_{\alpha}$ approves all candidates in $C'$ but at most $\ell-1$ candidates in $W$. Therefore, it holds that $(N'_{\alpha}, C, \ell)$ is a witness for an $\alpha$-EJR violation.
    
    For the reverse direction, let $(N'_{\alpha}, C', \ell)$ be an EJR violation. Thus, it holds that $\lvert N'_{\alpha} \rvert \geq \alpha \cdot  \ell \cdot \frac{n_\alpha}{k} = \ell \cdot \frac{n \cdot a \cdot b}{k \cdot b} = a \cdot \ell \cdot \frac{n}{k}$. 
    We now define the set $N' = \{ v \in N\colon \text{there exists an $i$ in 1 $\leq i \leq a$ with } v^{(i)}  \in N'_{\alpha} \}.$
    First, observe that $\lvert N'\rvert \geq \frac{\lvert N'_{\alpha} \rvert}{a} = \ell \cdot \frac{n}{k}$. 
    Furthermore, no voter in $N'$ approves more than $\ell-1$ candidates in $W$. Thus, $(N', C, \ell)$ witnesses an EJR violation for the original instance. This concludes the proof.   
\end{proof}

\ejrPlusAlphaFinding*
\begin{proof}
    The proof resembles the one from \Cref{proposition:JR_complexity_given_committee}. 
    We compute for every candidate $c \notin W$ and $\ell \in \{1, \dots, k \}$ the set of all voters that approve the candidate $c$ but approve at most $\ell-1$ candidates in $W$. Formally,
    \begin{equation*}
        N'_{c, \ell} := \{ i \in N_c\colon \lvert A_i \cap W \rvert < \ell \}.
    \end{equation*}
    Observe that this group of voters $N'_{c, \ell}$ forms an $\alpha$-EJR+ violation for $\alpha = \lvert N'_{c, \ell} \rvert \cdot \frac{k}{\ell n}$ since 
    $\lvert N'_{c, \ell} \rvert =  \alpha  \cdot \ell \cdot \frac{n}{k}$.
    
    Finally, let $\alpha = \max_{c \notin W, \ell \in \{1, \dots, k \}} \;  \lvert N'_{c, \ell} \rvert \cdot \frac{k}{\ell n}$.
    Therefore, $\alpha \leq \alpha_{\text{EJR+}}^W$. 
    
    It remains to show that $\alpha \geq \alpha_{\text{EJR+}}^W(I)$.
    By definition, there exists some $\alpha_{\text{EJR+}}^W(I)$-EJR+ violation $(N', c, \ell)$ for $N' \subseteq N$ and $c \in C$. 
    Thus, it holds that $\lvert N' \rvert \geq \alpha_{\text{EJR+}}^W(I) \cdot \ell \cdot \frac{n}{k}$. 
    Further, since the voters in $N'$ all approve $c$ but at most $\ell-1$ candidates in $W$, we get that $\lvert N' \rvert  \leq  \lvert \{i \in N_c\colon \lvert A_i \cap W \rvert < \ell \}\rvert \leq \lvert N'_{c, \ell} \rvert$. 
    Consequently,  $\lvert N'_{c, \ell} \rvert \geq \alpha_{\text{EJR+}}^W(I) \cdot \ell \cdot  \frac{n}{k}$ which implies that $\alpha_{\text{EJR+}}^W(I) \leq \lvert N'_{c, \ell} \rvert \cdot \frac{k}{\ell n} \leq \alpha$. 
    This concludes the proof.   
    The runtime is $O(mnk)$, since for every $c \notin W$ and $l \leq k$, we need to determine the number of voters who approve the candidate $ c$.
\end{proof}

\subsection{Restricted Domains} \label{section:appendix_restricted_domains}

\partydomainAlgorithm*
\label{appendix_proof_party_domain}
\begin{proof}
    We first show that $\alpha$-EJR+ and $\alpha$-EJR are equivalent for party-list profiles.
    The first direction follows from \Cref{prop:implications}, so we focus on the other one:
    Let $(N', c, \ell)$ be an EJR+ violation, i.e., $\lvert N' \rvert \geq \alpha \ell \frac{n}{k}$, $c \in (C \setminus W)\cap (\bigcap_{v\in N'}A_v)$ and for all $v \in N'$, we have $\lvert W \cap A_v \rvert < \ell$ and $\lvert\bigcap_{v \in N'}A_v \rvert \geq \ell$.
    This follows from the fact that all voters in $N'$ approve the same set of candidates; therefore, this also forms an $\alpha$-EJR violation.
    In the following, we will demonstrate the theorem for $\alpha$-EJR, which, as shown previously, also holds for $\alpha$-EJR+.
    
    Let $P$ be the set of parties.
    For $\alpha$-EJR, compute for every party $p \in P$, their support size denoted by $s_p$.
    Candidates are added to the committee $W$ sequentially in $k$ rounds.
    In each round $i$, add a candidate from party $p$ with
    \begin{equation*}
        \arg\max_{p \in P} \frac{s_p}{w_{p,i}+1}
    \end{equation*}
    where $w_{p,i}$ is the number of candidates approved by party $p$ in the committee $W$ at timestep $i$.
    The intuition behind the algorithm is the following:
    At every timestep, check which party causes the largest violation. 

    For the party $p$ not to get $(\alpha, w_{p,i}+1)$-cohesive, it has to hold that $\frac{n}{k} \cdot \alpha \cdot (w_{p,i} +1) > s_p$. Solving for $\alpha$, we get that $\alpha > \frac{s_p}{w_{p,i}+1} \cdot \frac{k}{n}$ where $\frac{k}{n}$ is a constant factor for all parties and can thus be ignored.

    We claim that this algorithm returns a committee with optimal $\alpha$-value with respect to EJR.
    Assume for the sake of contradiction that there does exist a committee $W'$ with a smaller $\alpha'$-value.
    If $W$ and $W'$ had the same number of candidates for each party, their $\alpha$-values would be equal.
    Thus, since $\lvert W \rvert = \lvert W'\rvert = k$, there must exist a party $p$ such that $w_p > w'_p$.
    Now, it holds  that $\alpha' \geq \frac{s_p}{w'_{p}+1} \cdot \frac{k}{n}$ due to party $p$.

    Next, consider the point where $p$ received its last candidate in the algorithm:
    At this specific timestep $i$ in the algorithm $\frac{s_p}{w_{p,i}+1} \geq \frac{s_t}{w_{t,i}+1}$ for all $t \in P$ and therefore it holds, that 
    \begin{equation*}
        \alpha \leq \frac{s_p}{w_{p,i}+1} \cdot \frac{k}{n} \leq \frac{s_p}{w'_{p}+ 1} \cdot \frac{k}{n} \leq \alpha'.
    \end{equation*}
    This contradicts our assumption that $\alpha' < \alpha$.

    In order to efficiently execute the algorithm we use a heap to store the parties.
    Then, in every one of the $k$ rounds, we need to compute the current value $\frac{s_p}{w_{p,i}+1}$ for every party $p \in P$.
    This can be implemented using a max-heap, yielding a runtime of  $O(\lvert P \rvert + k\log(\lvert P \rvert))$.

    For $\alpha$-JR we can find an optimal committee by making out the $k$ parties with the largest support and selecting a candidate from each of them. 
    In order to find $\alpha_{\text{JR}}^*(I)$ we consider the size of the $k+1$ largest party $P_{k+1}$ and compute $\alpha_{\text{JR}}^*(I) = \lvert P_{k+1}\rvert \cdot \frac{k}{n}$.
    Both can be done in linear time \citep{BLUM1973448}. 
\end{proof}

\viFindingEJRViolation* 
\begin{proof}
    Let $W$ be the committee, $N = \{v_1, \dots, v_n\}$ be the voters ordered according to the VI order $\sqsubset$, and $A$ be the approval profile.
    Furthermore, let $T(i,j)$ be the set of candidates that are commonly approved by voters $v_i$ and $v_j$. 
    Since each candidate is approved by a voter interval, we can conclude that every voter $v$ with $v_i \sqsubset v \sqsubset v_j$ approves every candidate in $T(i,j)$. \\
    For each $1\leq \ell \leq \lvert T(i,j) \rvert$, we compute the number of voters that approve at most $\ell-1$ many candidates in $W$. If this group is large enough to witness an $\alpha$-EJR violation, we know that $W$ does not satisfy $\alpha$-EJR since every voter in this group approves every candidate in $T(i,j)$.
    Finding a violation with this procedure, therefore, takes $O(nk)$ time for a given interval $(i,j)$.
    It is sufficient to consider every possible interval of which $O(n^2)$ many exist.
    Thus, the algorithm's runtime is in $O(n^3 k)$.

    Trivially, every violation we find with this procedure is an $\alpha$-EJR violation. In the following, we also prove that if there exists an $\alpha$-EJR violation $(N', C', \ell)$ then the algorithm above finds an $\alpha$-EJR violation $(N'', C'', \ell)$ with $N' \subseteq N''$.

    So let $N'$ be the set of voters witnessing the violation, and let $v_i$ be the first voter, as well as $v_j$ be the last voter in $N'$, according to the VI order. 
    When the algorithm considered the interval $(i,j)$, there were at least $\lvert N' \rvert \geq \frac{n}{k} \cdot \ell \cdot \alpha$ many voters approving at most $\ell-1$ many candidates in $W$. Further, each of them approves at least $\ell$ common candidates $C'$ with $C' \subseteq T(i,j)$. Thus, the algorithm finds a violation in the interval $(i,j)$. 
    This concludes the proof for EJR.
\end{proof}

\viDomainJR*
\begin{proof}
    To prove the theorem, we use \Cref{alg:vi_domain}. 
    Observe that the algorithm itself does not necessarily return a subset of size $k$. 
    By \Cref{lemma:VI_minimal_cardinality}, it holds that the algorithm returns the smallest subset $W$ that does not contain any $\alpha$-JR violation. Thus, we can execute \Cref{alg:vi_domain} and a committee satisfying $\alpha$-JR exists if and only if the algorithm returns a subset with at most $k$ candidates. 

    The runtime of the algorithm is $O(n^2m\log n)$ as in every one of the $n$ timesteps, we need to check if $\alpha$-JR is violated for a given committee $W$ and $\alpha$-value, which can be done in $O(nm)$ (shown in \Cref{proposition:JR_complexity_given_committee}).
    Using binary search over the distinct $\alpha$-values, of which at most $n$ exist due to \Cref{lemma:alphasToConsiderJR} requires an additional $\log n$ factor.
\end{proof}

\begin{lemma}\label{lemma:VI_minimal_cardinality}
    Let $W$ be the set of candidates returned by \Cref{alg:vi_domain}. 
    Then, there does not exist any set $W' \subseteq C$ with $\lvert W' \rvert < \lvert W\rvert$ that satisfies $\alpha$-JR.
\end{lemma}
\begin{proof}
    We prove this statement by induction over the voters, following the VI order.
    We say a subset of candidates is feasible if it satisfies $\alpha$-JR.
    Further, we say that a committee is \emph{optimal} if it is feasible and has minimal cardinality.
    Observe that the algorithm returns a feasible subset of candidates, so there always exists at least one optimal feasible subset.
    Let $C_j$ be the set of candidates selected up to time step $j$ by our algorithm. \\

    \noindent\textbf{Claim.} 
    For every $j\in\{1,\dots,n\}$, there exists an optimal subset of candidates $W^{(j)}$ such that every candidate that the greedy algorithm selected at some index $\leq j$ belongs to $W^{(j)}$, i.e., $C_j \subseteq W^{(j)}$. \\ 

    \noindent\textbf{Base:} $j=1$.
    If $C_j = \emptyset$, we are done, so assume that the set is non-empty.
    This means voter $v_1$ causes an $\alpha$-JR violation by herself, and every feasible subset of candidates must select a candidate that is approved by voter $v_1$.
    Observe that the candidate $c$ selected by \Cref{alg:vi_domain} is defined such that $r_{c'} \sqsubseteq r_c$ for all $c' \in A_{v_1}$. Since we are in the VI domain, this implies that $N_{c'} \subseteq N_{c}$. Therefore, if an optimal subset does not contain $c$, it has to contain such a $c'$, and we can swap them while upholding $\alpha$-JR. This proves the induction base. \\
    
    \noindent\textbf{Step:}
    Let the claim be satisfied for some $j-1 \geq 0$.
    If no candidate is selected in iteration $j$, we have $C_{j-1} = C_j$ and the claim follows from the induction hypothesis since $C_{j-1} \subseteq W^{(j-1)}$ and setting $W^{(j)} = W^{(j-1)}$. So assume a candidate $c_j$ is selected in iteration $j$. Further, if $c_j \in W^{(j-1)}$ then $C_{j} \subseteq W^{(j-1)}$ which proves the claim. So in the following assume that $c_j \not \in W^{(j-1)}$.

    Since there is an $\alpha$-JR violation caused by a set of voters $N''$ with $v \sqsubseteq v_{j}$ for every $v \in N''$ and containing voter $v_j$, $W^{(j-1)}\setminus C_{j-1}$ must contain a candidate $c''$ that is approved by some voter $v \in N''$ and thus $v \sqsubseteq v_j$.

    We now claim that removing $c''$ from $W^{(j-1)}$ and adding $c_j$ ends in a committee $W^{(j)}$ that also satisfies $\alpha$-JR. 

    So, assume for the sake of contradiction that this is not the case. Therefore, there exists a witness of a violation $(N',c)$.
    Further, since $W^{(j-1)}$ satisfies $\alpha$-JR there exists a voter $v'' \in N'$ that approves a candidate in $W^{(j-1)}$, but no candidate in $W^{(j)}$ which implies that $c'' \in A_{v''}$. 

    First, consider the case that $v_j \sqsubseteq v''$.
    Then, since $v \sqsubseteq v_j \sqsubseteq v''$ and $c'' \in A_{v''} \cap A_v$ it holds that $c'' \in A_{v_j}$. Therefore, by definition of $c_j$ we can conclude that $r_{c''} \sqsubseteq r_{c_j}$ which implies that every voter $v^*$ with $v_j \sqsubseteq v^*$ that approves $c''$ also approves $c_j$.
    Thus, voter $v''$ approves $c_j$ which contradicts $v'' \in N'$.

    It remains to consider the case that $v'' \sqsubset v_j$.
    Now, since by definition $C_{j-1}$ does not contain any $\alpha$-JR violation for $\{v_1, \dots, v_{j-1}\}$, and $C_{j-1} \subseteq W^{(j)}$ there has to exist a voter $v^* \in N'$ such that $v_j \sqsubseteq v^*$. But then $v'' \sqsubset v_j \sqsubseteq v^*$ which implies that $c \in A_{v_j}$. As in the first case, this implies that $r_{c} \sqsubseteq r_{c_j}$, and thus every voter $v^*$ with $v_j \sqsubseteq v^*$ who approves $c$ also approves $c_j$. But this is a contradiction to $v^* \in N'$.

    Consequently, $C_n$ is both minimal and feasible, corresponding to the output of our algorithm.
    Therefore, there cannot be any feasible committee $W'$ with $\lvert W' \rvert < \lvert W \rvert$.
\end{proof}

\ciFindingEJRViolation*
\begin{proof}
    We can use a similar proof as for the VI domain, but instead of considering every interval of voters, we consider every interval of candidates. 
    More precisely, denote by $N_{\ell}^{i,j} = \{v \in N: \{c_i, \dots, c_j\} \subseteq A_v, \lvert A_v\cap W \rvert < \ell \}$ for $1 \leq \ell \leq j-i+1$ the voters that approve every candidate in the interval but at most $\ell -1$ many candidates in $W$.
    If any such group $N_{\ell}^{i,j}$ is larger than $\alpha \ell \frac{n}{k}$ this group trivially witnesses an $\alpha$-EJR violation $(N_{\ell}^{i,j},\{c_i,\dots,c_j\} ,\ell)$. Since we can compute $N_{\ell}^{i,j}$ in $O(n)$ and we consider at most $k$ values for $\ell$ for at most $O(m^2)$ many pairs $i,j$ the runtime of the algorithm is at most $O(m^2kn)$.
    We claim that this procedure also always finds an $\alpha$-EJR violation whenever one exists.

    To prove this, let $(N', C', \ell)$ be any $\alpha$-EJR violation for the committee $W$. Further, let $c_i$ and $c_j$ be the first and last candidates in $C'$ according to the CI order. By definition, every voter in $N'$ approves at most $\ell-1$ many candidates in $W$ and $\lvert C' \rvert \geq \ell$. Thus, also $(N', \{c_i, \dots, c_j\}, \ell)$ is a witness of an EJR violation. When the algorithm considers the interval $i,j$, it therefore finds a violation $(N'',\{c_i,\dots, c_j\}, \ell)$ with $N' \subseteq N''$. This concludes the proof.    
   \end{proof}

\ciDomainJR*

\begin{proof} 
    We claim that \Cref{alg:ci_domain} returns a subset of size at most $k$ if and only if there exists a committee of size $k$ that satisfies $\alpha$-JR.
    First, observe that \Cref{alg:ci_domain} always outputs a subset of candidates that satisfies $\alpha$-JR by construction. 
    So if the output committee is of size at most $k$, there trivially also exists a committee satisfying $\alpha$-JR.
    It remains to show that if a committee of size $k$ exists, then the algorithm also finds one.

    We prove the statement inductively over the set of candidates $\{c_1, \dots, c_m\}$, going over them according to the linear order $\sqsubset$.
    $W$ is the committee returned by our algorithm.
    Let $C_j := W \cap \{c_1, \dots, c_j\}$ be the subset of candidates selected by the algorithm out of the first $j$ candidates according to the CI order. \\
    
    \noindent \textbf{Claim.}
    For a given $j$, there exists a minimal committee $W^{(j)}$ that satisfies $\alpha$-JR and $C_j \subseteq W^{(j)}$. \\
    
    \noindent \textbf{Base:} $j=1$.
    As discussed above, the subset of candidates returned by the algorithm satisfies $\alpha$-JR. Therefore, there exists a subset of candidates that satisfies $\alpha$-JR and thus also one of minimal cardinality. Now, if $c_1$ is selected by our algorithm, then $c_1$ must also be in every optimal committee since the set of candidates $\{c_2, \dots, c_m\}$ and therefore all possible subsets of size $k$ of this set would result in an $\alpha$-JR violation.
    Therefore, $C_1 \subseteq W^{(1)}$ for some optimal committee $ W^{(1)}$. \\

    \noindent \textbf{Step.} 
    Let the claim hold for $j-1 \geq 0$.
    We show that it also holds for $j$.
    If $c_j$ is not selected by the algorithm, then $C_j = C_{j-1}  \subseteq W^{(j-1)}$ by the induction hypothesis and we can conclude the claim by choosing $W^{(j)}= W^{(j-1)}$.
    Therefore, assume our algorithm selects $c_j$.
    Further, if $c_j \in W^{(j-1)}$ it holds that $C_j \subseteq W^{(j-1)}$ and again the claim follows.
    Since $c_j$ is selected by the algorithm, there is an $\alpha$-JR violation when considering $C_{j-1}\cup \{c_{j+1},\dots,c_m\}$ by a group of voters $N'$ approving a common candidate $c$ with $c \sqsubseteq c_j$ and no voter $v \in N'$ approves a candidate $c' \in \{c_{j+1}, \dots, c_m\}$.
    By assumption, $W^{(j-1)}$ does not have any violations, which means there is some candidate $c'' \in W^{(j-1)}$ 
    which is approved by at least one voter $v \in N'$ and $c'' \sqsubseteq c_j$. 
    Since we assume that $c_j \not \in W^{(j-1)}$ we can conclude that $c_j \neq c''$.
    Now, we denote as $W^{(j)}$ the committee after replacing $c''$ in $W^{(j-1)}$ with candidate $c_j$
    and claim that $W^{(j)}$ does not violate $\alpha$-JR.
    For this, assume for the sake of contradiction that there exists a cohesive group $(N'',c)$ that causes an $\alpha$-JR violation for $W^{(j)}$. This implies that no voter in $N''$ approves a candidate in $W^{(j)}$. First, consider the case $c_j \sqsubset  c$. Since $W^{(j-1)}$ satisfies  $\alpha$-JR  there exists a $v \in N''$ with $c^* \in A_v \cap W^{(j-1)}$ for some $c^*$. $c^* \in W^{(j)}$ would imply that $v$ approves a candidate in $W^{(j)}$. Thus, $c^* = c''$. But since $v$ approves $c$ and $c^*$ and $c^* \sqsubset c_j \sqsubset c$ this implies that $v$ approves $c_j$; a contradiction.

    Thus, in the following, we can assume that $c \sqsubset c_j$. 
    By definition of the algorithm, we know that the subset of candidates $C_{j-1} \cup \{c_{j}, \dots, c_m\}$ satisfies $\alpha$-JR.
    Therefore, there exists a voter $v \in N''$ that approves a candidate $c^* \in C_{j-1} \cup \{c_{j}, \dots, c_m\}$, but no candidate in $W^{(j)}$. Since $C_{j-1} \cup \{c_j\} \subseteq W^{(j)}$ it has to hold that $c_j \sqsubset c^*$. But then voter $v$ approves both $c$ and $c^*$, which together with $c \sqsubset c_j \sqsubset c^*$ implies that $c_j \in A_v$. 
    This contradicts our assumption that $v \in N''$ and thus concludes the proof.

    The runtime complexity of the proposed algorithm lies in $O(m^2n \log n)$.
    We need to check at most $n$ distinct $\alpha$-values according to \Cref{lemma:alphasToConsiderJR} and the algorithm iterates over all $m$ alternatives, checking in each step whether there is an $\alpha$-JR violation, which can be done in $O(mn)$.
    Using binary search for the $\alpha$-values results in an additional factor of $\log n$.
\end{proof}

\clearpage
\section{Further Empirical Insights} \label{section:appendix_experiments}
\subsection{Model Configuration}
For the experiments, we use the Python voting library by \citet{joss_abcvoting}. 
We generate approval ballots under two commonly used preference models: 
\emph{Impartial culture (IC)} and the \emph{Euclidean Threshold model (EUT)}. 
For both models, we vary the number of voters $n \in \{11, 12, 29, 59\}$, 
the number of candidates $m \in \{5, 9, 15\}$, and the committee size 
$k \in \{3, 5, 8, 11\}$. 
Under IC, each voter independently approves candidates with probability $p \in \{0.3, 0.5\}$. 
Under the Euclidean model, candidates and voters are embedded in the plane with voters sampled around the origin $(0,0)$ using a Gaussian distribution with $\sigma = 0.5$, and approvals are determined by a distance threshold $t \in \{1.7, 2.3\}$ (see \citep{joss_abcvoting} for more details). 
For each voting rule, we compute at most 5 committees. 
To compute each voting rule's performance, we compare the optimal $\alpha$-value with the mean $\alpha$-value of the returned committees.

\subsection{Further Results on $\alpha$-EJR}
We complement the results on $\alpha$-EJR presented in \Cref{section:experiments} with \Cref{table:ratio_EJR}. 
We separate the performance of the voting rules according to the parameters $k$, $m$, and $n$.
Every entry is averaged over 400 randomly generated instances. 
This comes from the fact that for every $(k, m, n)$-tuple, we have two sampling configurations for both sampling models. \\
PAV performs best among the considered methods.
MES and sequential Phragmén still perform reasonably well. 
For all combinations of parameters, the average additive distance between the $\alpha$-value achieved by the voting rules and the optimal $\alpha$-value is between 0.01 and 0.2, showing that these three rules are actually not as bad as the worst-case would suggest. 
This is especially surprising for sequential Phragmén, which does not satisfy 1-EJR. 
Unsurprisingly, CC performs poorly relative to the other rules, as its sole concern is maximizing the number of voters who approve a candidate. 
Nevertheless, CC often contains committees that also satisfy EJR, resulting in an overall moderate performance in approximating the optimal $\alpha$-value. \citep{bredereck2019experimental}. 
Therefore, we need to note that if one were not to consider the worst committees returned by CC, our analysis would have performed better.

\begin{table}[ht!]
\caption{Additive distance between $\alpha$-value achieved by the four different voting rules and the optimal achievable $\alpha$-value for $\alpha$-EJR.
This is averaged over 400 instances each and contains the IC and EUT model.}
\label{table:ratio_EJR}
\renewcommand{\arraystretch}{1.5}
\resizebox{\columnwidth}{!}{%
\begin{tabular}{rrr|p{1.3cm}p{1.3cm}p{1.3cm}p{1.3cm}}
\toprule
k & m & n & \multicolumn{1}{c}{\textbf{MES}} & \multicolumn{1}{c}{\textbf{seq-Phrag}} & \multicolumn{1}{c}{\textbf{CC}} & \multicolumn{1}{c}{\textbf{PAV}} \\
\midrule
3 & 5 & 11 & 0.010133 & 0.010473 & 0.055392 & 0.007784 \\
3 & 5 & 12 & 0.008507 & 0.007760 & 0.059616 & 0.008455 \\
3 & 5 & 29 & 0.017004 & 0.017047 & 0.034153 & 0.015237 \\
3 & 5 & 59 & 0.013083 & 0.013210 & 0.027110 & 0.012722 \\
3 & 9 & 11 & 0.019667 & 0.019884 & 0.101460 & 0.018074 \\
3 & 9 & 12 & 0.032172 & 0.033167 & 0.102177 & 0.028019 \\
3 & 9 & 29 & 0.028139 & 0.029149 & 0.049896 & 0.026893 \\
3 & 9 & 59 & 0.028475 & 0.028294 & 0.035160 & 0.027219 \\
3 & 15 & 11 & 0.043034 & 0.043642 & 0.159519 & 0.039949 \\
3 & 15 & 12 & 0.047342 & 0.046887 & 0.146009 & 0.045783 \\
3 & 15 & 29 & 0.035782 & 0.036655 & 0.069678 & 0.036067 \\
3 & 15 & 59 & 0.036084 & 0.036719 & 0.048757 & 0.035427 \\
5 & 9 & 11 & 0.038204 & 0.038030 & 0.180881 & 0.023542 \\
5 & 9 & 12 & 0.037530 & 0.034058 & 0.166367 & 0.027283 \\
5 & 9 & 29 & 0.041997 & 0.041559 & 0.081327 & 0.031221 \\
5 & 9 & 59 & 0.036419 & 0.037163 & 0.058316 & 0.034893 \\
5 & 15 & 11 & 0.044592 & 0.044854 & 0.275377 & 0.036544 \\
5 & 15 & 12 & 0.039661 & 0.039713 & 0.253620 & 0.034005 \\
5 & 15 & 29 & 0.055968 & 0.055875 & 0.118074 & 0.048463 \\
5 & 15 & 59 & 0.051038 & 0.050536 & 0.085691 & 0.045847 \\
8 & 9 & 11 & 0.011152 & 0.014509 & 0.158543 & 0.011192 \\
8 & 9 & 12 & 0.013481 & 0.016815 & 0.139576 & 0.009226 \\
8 & 9 & 29 & 0.017217 & 0.020142 & 0.073613 & 0.013810 \\
8 & 9 & 59 & 0.018103 & 0.018436 & 0.044746 & 0.014925 \\
8 & 15 & 11 & 0.046080 & 0.047793 & 0.321579 & 0.036526 \\
8 & 15 & 12 & 0.055936 & 0.054837 & 0.320223 & 0.036548 \\
8 & 15 & 29 & 0.050834 & 0.053788 & 0.170933 & 0.039221 \\
8 & 15 & 59 & 0.058181 & 0.059281 & 0.107438 & 0.053653 \\
11 & 15 & 11 & 0.039556 & 0.042946 & 0.298409 & 0.027728 \\
11 & 15 & 12 & 0.033503 & 0.038808 & 0.301895 & 0.027706 \\
11 & 15 & 29 & 0.054533 & 0.058232 & 0.175163 & 0.039851 \\
11 & 15 & 59 & 0.050336 & 0.053688 & 0.109126 & 0.043724 \\
\bottomrule
\end{tabular}}
\end{table}

\subsection{Plots for JR}
We present the same plots we considered for EJR for JR. 
The distribution of the $\alpha$-values is significantly closer towards $0$ than for EJR. 
The fact that the optimal value is $0$ in almost $80$ per cent of the instances in the IC model shows that there are many instances in which committees exist that satisfy every voter with at least one candidate.
\Cref{fig:ci_jr_cumulative_alpha} again shows the performance of the four voting rules we considered in this analysis.
CC is very close to the optimal $\alpha$-value for both the IC and EUT model. This is not surprising, since
CC maximizes the coverage, and therefore has an optimal $\alpha$-value of $0$ whenever $\alpha_{\text{JR}}^*$ is $0$. 

\begin{figure*}[t]  
    \centering
    \includegraphics[width=\textwidth]{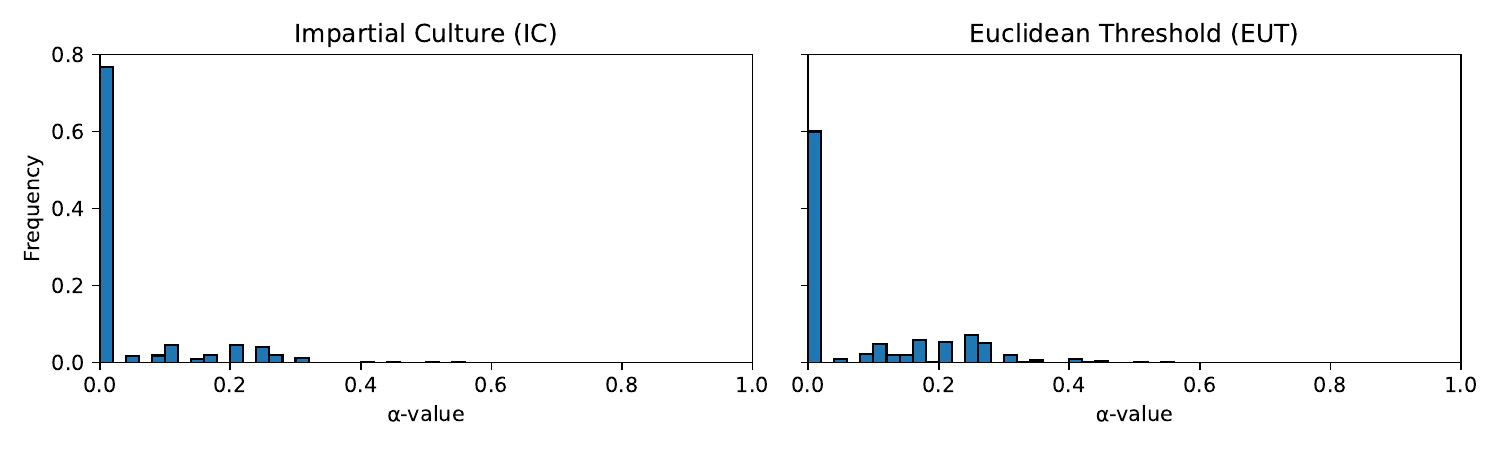}
    \caption{Distribution of optimal $\alpha$-values, $\alpha_{\text{JR}}^*$, under the Impartial Culture and the Euclidean Threshold model, based on 6400 generated instances each.}
    \label{fig:alpha_values_jr_overall}
\end{figure*}
\begin{figure*}[t]
    \centering
    \includegraphics[width=\textwidth]{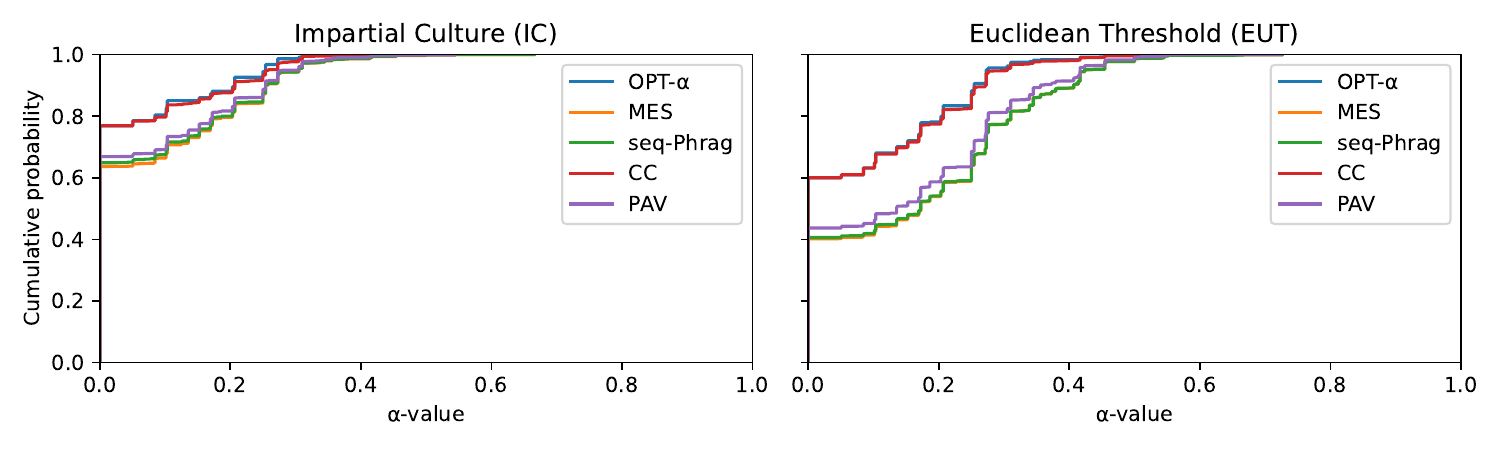}
    \caption{Cumulative distribution of the optimal $\alpha$-value, $\alpha_{\text{JR}}^*$,  and the $\alpha$-value achieved by MES, seq-Phragmén, CC, and PAV under the Impartial Culture and Euclidean Threshold model, based on 6400 samples each.}
    \label{fig:ci_jr_cumulative_alpha}
\end{figure*}

\end{document}